\documentclass[10pt, conference, letterpaper]{IEEEtran}
\usepackage{times,amsmath}
\usepackage{graphicx}
\usepackage{amsmath, amsthm, amssymb}
\usepackage{algorithm}
\usepackage[noend]{algorithmic}
\usepackage{comment}
\usepackage[caption=false,font=footnotesize]{subfig}

\newtheorem{theorem}{Theorem}
\newtheorem{lemma}{Lemma}

\newtheorem{definition}{Definition}
\newcommand\argmax{\operatorname{arg\,max}}

\newcommand{\mynoindent}{\vspace{0.05in}\noindent}

\title{Algorithmic Design for Competitive Influence Maximization Problems}

\author{
{Yishi Lin, John C.S. Lui}
\vspace{1.6mm}\\
\fontsize{10}{10}\selectfont\itshape
Department of Computer Science and Engineering,
The Chinese University of Hong Kong\\
\fontsize{9}{9}\selectfont\ttfamily\upshape
\{yslin,cslui\}@cse.cuhk.edu.hk
}

\begin{document}
\maketitle

\begin{abstract} 
Given the popularity of the viral marketing campaign in online social networks,
finding an effective method to identify a set of most influential nodes so to
compete well with other viral marketing competitors is of upmost importance.  We
propose a {\em``General Competitive Independent Cascade (GCIC)''} model to
describe the general influence propagation of two competing sources in the same
network.  We formulate the {\em ``Competitive Influence Maximization (CIM)''}
problem as follows: Under a prespecified influence propagation model and that
the competitor's seed set is known, how to find a seed set of $k$ nodes so as to
trigger the largest influence cascade?  We propose a general algorithmic
framework TCIM for the CIM problem under the GCIC model.  TCIM returns a
$(1-1/e-\epsilon)$-approximate solution with probability at least $1-n^{-\ell}$,
and has an efficient time complexity of $O(c(k+\ell)(m+n)\log n/\epsilon^2)$,
where $c$ depends on specific propagation model and may also depend on $k$ and
underlying network $G$.  To the best of our knowledge, this is the first general
algorithmic framework that has both $(1-1/e-\epsilon)$ performance guarantee and
practical efficiency.  We conduct extensive experiments on real-world datasets
under three specific influence propagation models, and show the efficiency and
accuracy of our framework.  In particular, we achieve up to four orders of
magnitude speedup as compared to the previous state-of-the-art algorithms with
the approximate guarantee.
\end{abstract}

\section{\bf Introduction} \label{sec: introduction}
With the popularity of online social networks (OSNs), viral marketing has become
a powerful method for companies to promote sales.  In 2003, Kempe et
al.~\cite{kempe2003maximizing} first formulated the influence maximization
problem: Given a network $G$ and an integer $k$, how to select a set of $k$
nodes in $G$ so that they can trigger the largest influence cascade under a
predefined influence propagation model.  The selected nodes are often referred
to as \textit{seed nodes}.  Kempe et al. proposed the \textit{Independent
Cascade (IC)} model and the \textit{Linear Threshold (LT)} model to describe the
influence propagation process.  They also proved that the influence maximization
problem under these two models is NP-hard and a natural greedy algorithm could
return $(1\!-\!1/e\!-\!\epsilon)$-approximate solutions for any $\epsilon$.
Recently, Tang et al.~\cite{tang2014influence} presented an algorithm with
$(1-1/e-\epsilon)$ approximation guarantee with probability at least
$1-n^{-\ell}$, and runs in time $O((\ell+k)(m+n)\log n/\epsilon^2)$. 

Recognizing that companies are competing in a viral marketing, a thread of work
studied the competitive influence maximization problem under a series of
competitive influence propagation models, where multiple sources spread the
information in a network simultaneously (e.g.,
\cite{bharathi2007competitive,carnes2007maximizing,borodin2010threshold}).  Many
of these work assumed that there are two companies competing with each other and
studied the problem from the ``follower's perspective''.  Here, the ``follower''
is the player who selects seed nodes with the knowledge that some nodes have
already been selected by its opponent.  For example, in the viral marketing, a
company introducing new products into an existing market can be regarded as the
follower and the set of consumers who have already purchased the existing
product can be treated as the nodes influenced by its competitor.  Briefly
speaking, the problem of {\em Competitive Influence Maximization (CIM)} is
defined as the following: Suppose we are given a network $G$ and the set of
\textit{seed nodes} selected by our competitor, how to select the \textit{seed
nodes} for our product in order to trigger the largest influence cascade?  These
optimization problems are NP-hard in general.  Therefore, the selection of seed
nodes is relied either on computationally expensive greedy algorithms with
$(1-1/e-\epsilon)$ approximation guarantee, or on heuristic algorithms with no
approximation guarantee.

To the best of our knowledge, for the CIM problem, there exists no algorithm
with both $(1-1/e-\epsilon)$ approximation guarantee and practical runtime
efficiency.  Furthermore, besides the existing models, we believe that there
will be more competitive influence propagation models proposed for different
applications in the future.  Therefore, we need a {\em general framework} that
can solve the competitive influence maximization problem under a variety of
propagation models. 

\mynoindent\textbf{Contributions}.
We make the following contributions:
\begin{itemize}

\item We define a \textit{General Competitive Independent Cascade (GCIC)} model
    and formally formulate the \textit{Competitive Influence Maximization (CIM)}
    problem.

\item For the CIM problem under a predefined GCIC model, we provide a
    \textit{Two-phase Competitive Influence Maximization (TCIM)} algorithmic
    framework generalizing the algorithm in \cite{tang2014influence}.  TCIM
    returns a $(1-1/e-\epsilon)$-approximate solution with probability at least
    $1-n^{-\ell}$, and runs in $O(c(\ell+k)(m+n)\log n/\epsilon^2)$, where $c$
    depends on specific propagation model, seed-set size $k$ and network $G$.

\item We analyze the performance of TCIM under three specific influence
    propagation models of the GCIC model as reported in literature
    \cite{carnes2007maximizing} and \cite{budak2011limiting}.

\item We conduct extensive experiments on real-world datasets to demonstrate the
    efficiency and effectiveness of TCIM.  In particular, when $k\!=\!50$,
    $\epsilon\!=\!0.5$ and $\ell\!=\!1$, TCIM returns solutions comparable with
    those returned by the previous state-of-the-art greedy algorithms, but TCIM
    runs {\em up to four orders of magnitude faster}.

\end{itemize}

This is the outline of our paper.  Background and related work are given in
Section \ref{sec: related}.  We define the \textit{General Competitive
Independent Cascade} model and the \textit{Competitive Influence Maximization}
problem in Section \ref{sec: problem}.  We present the TCIM framework in Section
\ref{sec: solution} and analyze the performance of TCIM under various influence
propagation models in Section \ref{sec: application}.  We compare TCIM with the
greedy algorithm with performance guarantee in Section \ref{sec: comparison},
and show the experimental results in Section \ref{section: experiments}.
Section \ref{sec: conclusion} concludes.

\section{\bf Background and Related Work} \label{sec: related}

\mynoindent\textbf{Single Source Influence Maximization.} In the seminal
work~\cite{kempe2003maximizing}, Kempe et al. proposed the \textit{Independent
Cascade (IC) model} and the \textit{Linear-Threshold (LT) model} and formally
defined the \textit{influence maximization problem}.  In the IC model, a network
$G$ is given as $G\!=\!(V,E)$ and each edge $e_{uv}\!\in\! E$ is associated with
a probability $p_{uv}$.  Initially, a set of nodes $S$ are \textit{active} and
$S$ is often referred to as the \textit{seed nodes}.  Each active node $u$ has a
single chance to influence its inactive neighbor $v$ and succeeds with
probability $p_{uv}$.  Let $\sigma(S)$ be the expected number of nodes $S$ could
activate, the \textit{influence maximization problem} is defined as how to
select a set of $k$ nodes such that $\sigma(S)$ is maximized.  This problem,
under both the IC model and the LT model, is NP-hard.  However, Kempe et
al.~\cite{kempe2003maximizing} showed that if $\sigma(S)$ is a monotone and
submodular function of $S$, a greedy algorithm can return a solution within a
factor of $(1-1/e-\epsilon)$ for any $\epsilon>0$, in polynomial time.  The
research on this problem went on for around ten years (e.g.,
\cite{leskovec2007cost,chen2009efficient,chen2010scalable,
chen2010scalableicdm,goyal2011celf++,jung2012irie}), but it is not until very
recently, that Borgs et al.~\cite{borgs2014maximizing} made a breakthrough and
presented an algorithm that simultaneously maintains the performance guarantee
and significantly reduces the time complexity.  Recently, Tang et
al.~\cite{tang2014influence} further improved the method in
\cite{borgs2014maximizing} and presented an algorithm TIM/TIM$^+$, where
\textit{TIM} stands for \textit{Two-phase Influence Maximization}.  It returns a
$(1-1/e-\epsilon)$-approximate solution with probability at least $1-n^{-\ell}$
and runs in time $O((\ell+k)(m+n)\log n/\epsilon^2)$, where $n=|V|$ and $m=|E|$.

\mynoindent\textbf{Competitive Influence Maximization.} We review some work that
modeled the competition between two sources and studied the influence
maximization problem from the ``{\em follower's perspective}''.  In general, the
majority of these works considered competition between two players (e.g., two
companies), and the ``follower'' is the player who selects the set of
\textit{seed nodes} with the knowledge of the seed nodes selected by its
competitor.  In \cite{carnes2007maximizing}, Carnes et al. proposed the
\textit{Distance-based model} and the \textit{Wave propagation model} to
describe the influence spread of competing products and considered the influence
maximization problem from the follower's perspective.  Bharathi et
al.~\cite{bharathi2007competitive} proposed an extension of the single source IC
model and utilized the greedy algorithm to compute the best response to the
competitor.  Motivated by the need to limit the spread of rumor in the social
networks, there is a thread of work focusing on how to maximize rumor
containment (e.g., \cite{kostka2008word, budak2011limiting, he2012influence}).
For example, Budak et al.~\cite{budak2011limiting} models the competition
between the ``bad'' and ``good'' source.  They focused on minimizing the number
of nodes end up influenced by the ``bad'' source.

\section{{\bf Competitive Influence Maximization Problem}} \label{sec: problem}

In this section, we first introduce the ``\textit{General Competitive
Independent Cascade (GCIC)}'' model which models the influence propagation of
two competing sources in the same network.  Based on the GCIC model, we then
formally define the \textit{Competitive Influence Maximization (CIM)} problem.

\subsection{General Competitive Independent Cascade Model} \label{sec: general
cicm} Let us first define the \textit{General Competitive Independent Cascade
(GCIC)} model.  A social network can be modeled as a directed graph
$G\!=\!(V,E)$ with $n\!=\!|V|$ nodes and $m\!=\!|E|$ edges.  Users in the social
network are modeled as nodes while directed edges between nodes represent the
interaction between users.  A node $v$ is a neighbor of node $u$ if there is an
edge from $u$ to $v$ in $G$.  Every edge $e_{uv}\in E$ is associated with a
length $d_{uv}>0$ and a probability $p_{uv}$ denoting the influence node $u$ has
on $v$.  For $e_{uv}\notin E$, we assume $p_{uv}=0$ and $d_{uv}=+\infty$.  For
the ease of presentation, we assume the length of all edges is $1$.  Our
algorithm and analysis can be easily extended to the case where edges have
nonuniform lengths.

Denote source $A$ and source $B$ as two sources that simultaneously spread
information in the network $G$.  A node $v\in V$ could be in one of these three
states: $S$, $I_A$ and $I_B$.  Nodes in state \textit{S}, the
\textit{susceptible state}, have not been influenced by any source.  Nodes in
state $I_A$ (resp. $I_B$) are influenced by source $A$ (resp. $B$).  Once a node
becomes influenced, it cannot change its state.  Initially, source $A$ and
source $B$ can each specify a set of seed nodes, which we denote as
$S_A\subseteq V$ and $S_B\subseteq V$.  We refer to nodes in $S_A$ (resp. $S_B$)
as \textit{seeds} or \textit{initial adopters} of source $A$ (resp. $B$).
Following the previous work that modeled the competitive influence propagation
(e.g., \cite{carnes2007maximizing,budak2011limiting}), we also assume $S_A\cap
S_B=\emptyset$.

As in the single source \textit{Independent Cascade (IC)} model, an influenced
node $u$ influences its neighbor $v$ with probability $p_{uv}$ and we say each
edge $e_{uv}\in E$ is \textit{active} with probability $p_{uv}$.  We can first
determine the set of \textit{active} edges $E_a\subseteq E$ by generating a
random number $\l_{uv} \! \in\!  [0,1]$ for every edge $e_{uv}\in E$, and select
$e_{uv}$ when $l_{uv} \! \leq \!  p_{uv}$.  Let $d_{E_a}(v,u)$ be the shortest
distance from $v$ to $u$ through edges in $E_a$ and assume
$d_{E_a}(v,u)=+\infty$ if $v$ cannot reach $u$ through active edges.  Moreover,
let $d_{E_a}(S_A\cup S_B, u)=\min_{v\in S_A\cup S_B}d_{E_a}(v,u)$ be the
shortest distance from nodes in $S_A\cup S_B$ to node $u$ through edges in
$E_a$.  For a given $E_a$, we say a node $v$ is a \textit{nearest initial
adopter} of $u$ if $v\!\in\!S_A\! \cup \! S_B$ and $d_{E_a}(v,
u)\!=\!d_{E_a}(S_A\cup S_B, u)$.  In the GCIC model, for a given $E_a$, a node
$u$ will be in the same state as that of one of its \textit{nearest initial
adopters} at the end of the influence propagation process.  The expected
influence of $S_B$ is the expected number of nodes in state $I_B$ at the end of
the influence propagation process, where the expectation is taken over the
randomness of $E_a$.  Specific influence propagation model of the GCIC model
will specify how the influence propagates in detail, including the tie-breaking
rule for the case where both nodes in $S_A$ and $S_B$ are nearest initial
adopters of a node.

Moreover, we make the following assumptions about the GCIC model.  Given $S_A$,
let $\sigma_u(S_B|S_A)$ be conditional probability that node $u$ will be
influenced by source $B$ when $S_B$ is used as the seed set for source $B$.  We
assume that $\sigma_u(S_B|S_A)$ is a monotone and submodular function of
$S_B\subseteq V\backslash S_A$ for all $u\in V$.  Formally, for any seed set
$S_A\subseteq V$, $S_{B_1}\subseteq S_{B_2}\subseteq V\backslash S_A$ and node
$v\in V\backslash(S_A\cup S_{B_2})$, we have $\sigma_u(S_{B_1}|S_A)\leq
\sigma_u(S_{B_2}|S_A)$ and $\sigma_u(S_{B_1}\cup
\{v\}|S_A)-\sigma_u(S_{B_1}|S_A)\geq \sigma_u(S_{B_2}\cup
\{v\}|S_A)-\sigma_u(S_{B_2}|S_A)$ hold for all $u\in V$.  Let $\sigma(S_B|S_A)$
be the expected influence of $S_B$ given $S_A$, because
$\sigma(S_B|S_A)=\sum_{u\in V}\sigma_u(S_B|S_A)$, $\sigma(S_B|S_A)$ is also a
monotone and submodular function of $S_B\subseteq V\backslash S_A$.

We call this the \textit{General Competitive Independent Cascade} model because
for any given graph $G=(V,E)$ and $S_B\subseteq V$, the expected influence of
$S_B$ given $S_A=\emptyset$ equals to the expected influence of $S_B$ in the
single source IC model.  Note that there are some specific instances of the GCIC
model, for example, the {\em Distance-based Model} and the {\em Wave propagation
Model} \cite{carnes2007maximizing}.  We will elaborate on them in later
sections.

\subsection{Problem Definition}

Given a directed graph $G$, a specific instance of the \textit{General
Competitive Independent Cascade} model (e.g., the Distance-based Model), and
seeds $S_A$ for source $A$, let us formally define the \textit{Competitive
Influence Maximization} problem.
\begin{definition}[Competitive Influence Maximization Problem]
Suppose we are given a specific instance of the \textit{General Competitive
Independent Cascade} model (e.g., the Distance-based Model), a graph $G=(V,E)$
and the seed set $S_A\!\subseteq\! V$ for source $A$, find a set $S_B^*$ of $k$
nodes for source $B$ such that the expected influence of $S_B^*$ given $S_A$ is
maximized, i.e.,
\begin{equation}
S_B^*=\argmax_{S_B\in\{S\subseteq V\backslash S_A,\ |S|=k\}}\sigma(S_B|S_A).
\end{equation}
\end{definition}
For the above problem, we assume $|V\backslash S_A|\geq k$.  Otherwise, we can
simply select all nodes in $V\backslash S_A$.  The \textit{Competitive Influence
Maximization} (CIM) problem is NP-hard in general.  In this paper, our goal is
to provide an approximate solution to the CIM problem with an approximation
guarantee and at the same time, with practical run time complexity.

\section{{\bf Proposed Solution Framework to the CIM Problem}}\label{sec: solution}

In this section, we present the \textit{Two-phase Competitive Influence
Maximization (TCIM)} algorithm to solve the \textit{Competitive Influence
Maximization} problem.  We extend the TIM/TIM$^+$
algorithm~\cite{tang2014influence}, which is designed for the single source
influence maximization problem, to a general framework for the CIM problem under
any specific instance of the \textit{General Competitive Independent Cascade}
model, while maintaining the $(1\!-\!1/e\!-\!\epsilon)$ approximation guarantee
and practical efficiency.

Let us first provide some basic definitions and give the high level idea of the
TCIM.  Then, we provide a detailed description and analysis of the two phases of
the TCIM algorithm, namely the \textit{Parameter estimation and refinement}
phase, and the \textit{Node selection} phase.

\subsection{Basic definitions and high level idea}

Motivated by the definition of ``RR sets'' in \cite{borgs2014maximizing} and
\cite{tang2014influence}, we define the \textit{Reverse Accessible Pointed Graph
(RAPG)}.  We then design a \textit{scoring system} such that for a large number
of random RAPG instances and given seed sets $S_A$ and $S_B$, the average score
of $S_B$ for each RAPG instance is a good approximation of the expected
influence of $S_B$ given $S_A$.

Let $d_g(u,v)$ be the shortest distance from $u$ to $v$ in a graph $g$ and
assume $d_g(u,v)=+\infty$ if $u$ cannot reach $v$ in $g$.  Let $d_g(S,v)$ be the
shortest distance from nodes in set $S$ to node $v$ through edges in $g$, and
assume $d_g(S,v)=+\infty$ if $S=\emptyset$ or $S\neq\emptyset$ but there are no
paths from nodes in $S$ to $v$.  We define the \textit{Reverse Accessible
Pointed Graph (RAPG)} and the random RAPG instance as the following.

\begin{definition}[Reverse Accessible Pointed Graph]
For a given node $v$ in $G$ and a subgraph $g$ of $G$ obtained by removing each
edge $e_{uv}$ in $G$ with probability $1-p_{uv}$, let $R=(V_R,E_R)$ be the
Reverse Accessible Pointed Graph (RAPG) obtained from $v$ and $g$.  The node set
$V_R$ contains $u\in V$ if $d_g(u,v)\leq d_g(S_A,v)$.  And the edge set $E_R$
contains edges on all shortest paths from nodes in $V_R$ to $v$ through edges in
$g$.  We refer to $v$ as the ``root'' of $R$.
\end{definition}

\begin{definition}[Random RAPG instance]
Let $\mathcal{G}$ be the distribution of $g$ induced by the randomness in edge
removals from $G$. A random RAPG instance $R$ is a Reverse Accessible Pointed
Graph (RAPG) obtained from a randomly selected node $v\in V$ and an instance of
$g$ randomly sampled from $\mathcal{G}$.
\end{definition}

Figure \ref{fig: RAPG} shows an example of a random RAPG instance $R=(V_R,E_R)$
with $V_R=\{2,5,7,9,10,11\}$ and $E_R=\{e_4,e_5,e_7,e_8,e_9,e_{10},e_{11}\}$.
The ``root'' of $R$ is node $v=2$.

\begin{figure}[htb]
\centering
\includegraphics[width=0.3\textwidth]{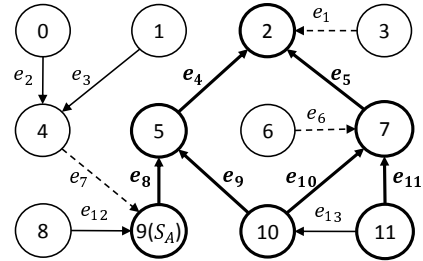} 
\caption{Example of a random RAPG instance: The graph $G=(V,E)$ contains 
$12$ nodes and $13$ directed edges each represented by an arrow. 
The random subgraph $g$ is obtained from $G$ by removing $3$ directed edges 
represented by dashed arrows, i.e., $e_1$, $e_6$ and $e_7$. 
From $g$ and the randomly selected ``root'' node $2$, 
we get the random RAPG instance $R=(V_R,E_R)$ where $V_R=\{2,5,7,9,10,11\}$ 
and $E_R=\{e_4,e_5,e_8,e_9,e_{10},e_{11}\}$.}\label{fig: RAPG}
\end{figure}

Now we present the \textit{scoring system}.  For a random RAPG instance
$R=(V_R,E_R)$ obtained from $v$ and $g\sim\mathcal{G}$, the score of a node set
$S_B$ in $R$ is defined as follows.

\begin{definition}[Score]\label{def: score}
Suppose we are given a random RAPG instance $R=(V_R,E_R)$ obtained from $v$ and
$g\sim\mathcal{G}$. The score of a node set $S_B$ in $R$, denoted by
$f_R(S_B|S_A)$, is defined as the probability that node $v$ will be influenced
by source $B$ when 1) the influence propagates in graph $g$ with all edges being
``active''; and 2) $S_A\cap V_R$ and $S_B\cap V_R$ are seed sets for source $A$
and $B$.
\end{definition}

Recall that for the \textit{General Competitive Independent Cascade} model, we
assume that for any node $u\in V$, the conditional probability
$\sigma_u(S_B|S_A)$ is a monotone and submodular function of $S_B\subseteq
V\backslash S_A$. It follows that, for any given $S_A$ and $R$, $f_R(S_B|S_A)$
is also a monotone and submodular function of $S_B\subseteq V\backslash S_A$.
Furthermore, we define the marginal gain of the score as follows.

\begin{definition}[Marginal gain of score]
For a random RAPG instance $R$ with root $v$, we denote 
\begin{equation}\label{eq: def marginal gain}
\Delta_R(w|S_A,S_B)=f_R(S_B\cup \{w\}|S_A)-f_R(S_B|S_A)
\end{equation}
as the marginal gain of score if we add $w$ to the seed set $S_B$.
\end{definition}

From the definition of GCIC model and that of the RAPG, we know that for any
RAPG instance $R$ obtained from $v$ and $g$, $R$ contains all nodes that can
possibly influence $v$ and all shortest paths from these nodes to $v$.  Hence,
for any given $S_A$, $S_B$ and node $w$, once an instance $R$ is constructed,
the evaluation of $f_R(S_B|S_A)$ and $\Delta_R(w|S_A,S_B)$ can be done based on
$R$ without the knowledge of $g$.  

From Definition \ref{def: score}, for any $S_B\subseteq V\backslash S_A$, the
expected value of $f_R(S_B|S_A)$ over the randomness of $R$ equals to the
probability that a randomly selected node in $G$ can be influenced by $S_B$.
Formally, we have the following lemma.

\begin{lemma}\label{lemma: EISB}
For given seed set $S_A$ and $S_B$, we have
\begin{equation}
\sigma(S_B|S_A)=n\cdot \mathbb{E}[f_R(S_B|S_A)]
\end{equation}
where the expectation of $n\cdot \mathbb{E}[f_R(S_B|S_A)]$ is taken over the
randomness of $R$, and $n$ is the number of nodes in $G$, or $n=|V|$.
\end{lemma}

Now we provide the Chernoff-Hoeffding bound in the form that we will frequently
use throughout this paper.

\begin{lemma}[Chernoff-Hoeffding Bound]\label{lemma: hoeffding}
Let $X$ be the summation of $\theta$ i.i.d. random variables 
bounded in $[0,1]$ with a mean value $\mu$.
Then, for any $\delta>0$,
\begin{align}
\Pr[X>(1+\delta)\theta\mu]&\leq 
\exp\Big(-\frac{\delta^2}{2+\delta}\cdot \theta\mu\Big),\\
\Pr[X<(1-\delta)\theta\mu]&\leq 
\exp\Big(-\frac{\epsilon^2}{2}\cdot \theta\mu\Big).
\end{align}
\end{lemma}

By Lemma \ref{lemma: EISB} and Chernoff-Hoeffding bound, for a sufficiently
large number of random RAPG instances, the average score of a set $S_B$ in those
RAPG instances could be a good approximation to the expected influence of $S_B$
in $G$.  The main challenge is how to determine the number of RAPG required, and
how to select seed nodes for source $B$ based on a set of random RAPG instances.
Similar to the work in \cite{tang2014influence}, TCIM consists of two phases as
follows.

\begin{enumerate}
\item \textit{Parameter estimation and refinement}:
Suppose $S_B^*$ is the optimal solution to the 
\textit{Competitive Influence Maximization Problem} and let 
$OPT=\sigma(S_B^*|S_A)$ be the expected influence of $S_B^*$ given $S_A$.
In this phase, TCIM estimates and refines a lower bound of $OPT$ and 
uses the lower bound to derive a parameter $\theta$.

\item \textit{Node selection}: In this phase, TCIM first generates a set 
$\mathcal{R}$ of $\theta$ random RAPG instances of $G$, where $\theta$ is a 
sufficiently large number obtained in the previous phase.
Using the greedy approach,
TCIM returns a set of seed nodes $S_B$ for source $B$
with the goal of maximizing $\sum_{R\in\mathcal{R}}f_R(S_B|S_A)$.
\end{enumerate}

\subsection{Node Selection}

Algorithm \ref{algo: node selection} shows the pseudo-code of the \textit{node
selection} phase.  Given a graph $G$, the seed set $S_A$ of source $A$, the seed
set size $k$ for source $B$ and a constant $\theta$, the algorithm returns a
seed set $S_B$ of $k$ nodes for source $B$ with a large influence spread.  In
Line \ref{algo: node selection start}-\ref{algo: node selection: initmg}, the
algorithm generates $\theta$ random RAPG instances and initializes
$MG_\mathcal{R}(u):=\sum_{R\in\mathcal{R}}f_R(\{u\}|S_A)$ for all nodes $u\in
V\backslash S_A$.  Then, in Line \ref{algo: node selection: ite s} - \ref{algo:
node selection: ite e}, the algorithm selects seed nodes $S_B$ iteratively using
the greedy approach with the goal of maximizing
$\sum_{R\in\mathcal{R}}f_R(S_B|S_A)$.
\begin{algorithm}[htb]
\small
\caption{NodeSelection $(G,S_A,k,\theta)$}
\label{algo: node selection}
\begin{algorithmic}[1]
	\STATE Generate a set $\mathcal{R}$ of $\theta$ random RAPG instances. 
		\label{algo: node selection start}  
	\STATE Let $MG_\mathcal{R}(u)=\sum_{R\in \mathcal{R}}f_R(\{u\}|S_A)$ 
		for all $u\in V\backslash S_A$. \label{algo: node selection: initmg}
	\STATE Initialize the seed set $S_B=\emptyset$.\label{algo: node selection: ite s}
	\FOR{$i=1$ \TO $k$}
		\STATE Identity the node $v_i\!\in\!V\backslash(S_A\!\cup\!S_B)$ 
				with largest $MG_\mathcal{R}(v_i)$.
		\STATE Add $v_i$ to $S_B$.
		\IF {$i<k$}
			\STATE // Update $MG_\mathcal{R}(u)$ as 
					$\sum_{R\in\mathcal{R}}\Delta_R(u|S_A,S_B)$ 
			\STATE // for all $u\in V\backslash(S_A \cup S_B)$.
			\STATE Let $\mathcal{R'}\!=\!\{R|R\in\mathcal{R}, 
					\Delta_R(v_i|S_A,S_B\backslash\{v_i\})>0\}$.
					\label{algo: node selection: mgs}
			\FORALL {$R\in\mathcal{R}'$ and $u\in V_R\backslash(S_A\cup S_B)$} 
			\STATE $MG_\mathcal{R}(u)=MG_\mathcal{R}(u)
						-\Delta_R(u|S_A,S_B\backslash\{v_i\}).$
			\STATE $MG_\mathcal{R}(u)=MG_\mathcal{R}(u)+\Delta_R(u|S_A,S_B)$.
			\ENDFOR \label{algo: node selection: mge}
		\ENDIF
	\ENDFOR  \label{algo: node selection: ite e}
	\RETURN $S_B$ \label{algo: node selection end}
\end{algorithmic}
\end{algorithm}

\mynoindent\textbf{Generation of RAPG instances.} 
We adapt the randomized breadth-first search used in Borg et al.'s
method~\cite{borgs2014maximizing} and Tang et al.'s
algorithm~\cite{tang2014influence} to generate random RAPG instances.  We first
randomly pick a node $r$ in $G$.  Then, we create a queue containing a single
node $r$ and initialize the RAPG instance under construction as
$R=(V_R=\{r\},E_R=\emptyset)$.  For all $u\in V$, let $d_R(u,r)$ be the shortest
distance from $u$ to $r$ in the current $R$ and let $d_R(u,r)=+\infty$ if $u$
cannot reach $r$ in $R$.  We iteratively pop the node $v$ at the top of the
queue and examine its incoming edges.  For each incoming neighbor $u$ of $v$
satisfying $d_R(u,r)\geq d_R(v,r)+1$, we generate a random number $l\!\in
\![0,1]$.  With probability $p_{uv}$ (i.e. $l\leq p_{uv}$), we insert $e_{uv}$
into $R$ and we push node $u$ into the queue if it has not been pushed into the
queue before.  If we push a node $u$ of $S_A$ into the queue while examining the
incoming edge of a node $v$ with $d_R(v,r)=d$, we terminate the breadth-first
search after we have examined incoming edges of all nodes whose distance to $r$
in $R$ is $d$.  Otherwise, the breadth-first search terminates naturally when
the queue becomes empty.  If reverse the direction of all edges in $R$, we
obtain an accessible pointed graph with ``root'' $r$, in which all nodes are
reachable from $r$.  For this reason, we refer to $r$ as the ``root'' of $R$.

\mynoindent\textbf{Greedy approach.}
Let $F_\mathcal{R}(S_B|S_A)=\sum_{R\in\mathcal{R}}f_R(S_B|S_A)$ for all
$S_B\subseteq V\backslash S_A$.  Line \ref{algo: node selection: ite
s}-\ref{algo: node selection: ite e} in Algorithm \ref{algo: node selection}
uses the greedy approach to select a set of nodes $S_B$ with the goal of
maximizing $F_\mathcal{R}(S_B|S_A)$.  Since the function $f_R(S_B|S_A)$ is a
monotone and submodular function of $S_B\subseteq V\backslash S_A$ for any RAPG
instance $R$, we can conclude that $F_\mathcal{R}(S_B|S_A)$ is also a monotone
and submodular function of $S_B\subseteq V\backslash S_A$ for any $\mathcal{R}$.
Hence, the greedy approach in Algorithm \ref{algo: node selection} could return
a $(1-1/e)$ approximation solution~\cite{nemhauser1978analysis}.  Formally, let
$S_B^*$ be the optimal solution, the greedy approach returns a solution $S_B$
such that $F_\mathcal{R}(S_B|S_A)\geq (1-1/e)F_\mathcal{R}(S_B^*|S_A)$.

\mynoindent\textbf{The ``marginal gain vector''. } During the greedy selection
process, we maintain a vector $MG_\mathcal{R}$ such that
$MG_\mathcal{R}(u)=F_\mathcal{R}(S_B\cup\{u\}|S_A)-F_\mathcal{R}(S_B|S_A)$ holds
for current $S_B$ and all $u\in V\backslash(S_A\cup S_B)$.  We refer to
$MG_\mathcal{R}$ as the ``{\em marginal gain vector}''.  The initialization of
$MG_\mathcal{R}$ could be done during or after the generation of random RAPG
instances, whichever is more efficient.  At the end of each iteration of the
greedy approach, we update $MG_\mathcal{R}$.  Suppose in one iteration, we
expand the previous seed set $S_B'$ by adding a node $v_i$ and the new seed set
is $S_B=S_B'\cup\{v_i\}$.  For any RAPG instance $R$ such that $v_i\notin V_R$,
we would have $\Delta_R(u|S_A,S_B')=\Delta_R(u|S_A,S_B)$ for all $u\in
V\backslash(S_A\cup S_B)$.  And for any RAPG instance $R$ such that
$f_R(S_B'|S_A)=1$, for all $u\in V\backslash(S_A\cup S_B)$, we would have
$\Delta_R(u|S_A,S_B')=0$ and the marginal gain of score cannot be further
decreased.  To conclude, for a given RAPG instance $R=(V_R,E_R)$ and a node
$u\in V_R\backslash(S_A\cup S_B)$, $\Delta_R(u|S_A,S_B)$ differs from
$\Delta_R(u|S_A,S_B')$ only if $v_i\in V_R$ and $f_R(S_B'|S_A)<1$.  Hence, to
update $MG_\mathcal{R}(u)$ as $\sum_{R\in\mathcal{R}}\Delta_R(u|S_A,S_B)$ for
all $u\in V\backslash(S_A\cup S_B)$, it is not necessary to compute
$\Delta_R(u|S_A,S_B)$ for all $R\in\mathcal{R}$ and $u\in V\backslash(S_A\cup
S_B)$.  Note that for any RAPG instance $R$,
$\Delta_R(v_i|S_A,S_B\backslash\{v_i\})>0$ implies $v_i\in V_R$ and
$f_R(S_B'|S_A)>1$.  Therefore, Line \ref{algo: node selection: mgs}-\ref{algo:
node selection: mge} do the update correctly.

\mynoindent\textbf{Time complexity analysis.} Let $\mathbb{E}[N_R]$ be the
expected number of random numbers required to generate a random RAPG instance,
the time complexity of generating $\theta$ random RAPG instances is
$O(\theta\cdot \mathbb{E}[N_R])$.  Let $\mathbb{E}\left[|E_R|\right]$ be the
expected number of edges in a random RAPG instance, which is no less than the
expected number of nodes in a random RAPG instance.  We assume that the
initialization and update of $MG_\mathcal{R}$ takes time $O(c\theta\cdot
\mathbb{E}\left[|E_R|\right])$.  Here, $c=\Omega(1)$ depends on specific
influence propagation model and may also depend on $k$ and $G$.  In each
iteration, we go through $MG_\mathcal{R}(u)$ for all nodes $u\in
V\backslash(S_A\cup S_B)$ and select a node with the largest value, which takes
time $O(n)$.  Hence, the total running time of Algorithm \ref{algo: node
selection} is $O(kn+\theta\cdot \mathbb{E}[N_R] + c\theta\cdot
\mathbb{E}\left[|E_R|\right])$.  Moreover, from the fact that
$\mathbb{E}\left[|E_R|\right]\leq \mathbb{E}[N_R]$ and $c=\Omega(1)$, the total
running time can be written in a more compact form as 
\begin{equation}\label{eq: algo 1 time}
O(kn+c\theta\cdot \mathbb{E}[N_R]).
\end{equation}
In Section \ref{sec: application}, we will show the value of $c$ and provide the
total running time of the TCIM algorithm for several influence propagation
models.

\mynoindent\textbf{The Approximation guarantee.} From Lemma \ref{lemma: EISB},
we see that the larger $\theta$ is, the more accurate is the estimation of the
expected influence.  The key challenge now becomes how to determine the value of
$\theta$, i.e., the number of RAPG instances required, so to achieve certain
accuracy of the estimation.  More precisely, we would like to find a $\theta$
such that the node selection algorithm returns a
$(1\!-\!1/e\!-\!\epsilon)$-approximation solution.  At the same time, we also
want $\theta$ to be as small as possible since it has the direct impact on the
running time of Algorithm \ref{algo: node selection}.

Using the Chernoff-Hoeffding bound, the following lemma shows that for a set
$\mathcal{R}$ of sufficiently large number of random RAPG instances,
$F_\mathcal{R}(S_B|S_A)\cdot n/\theta
=\left(\sum_{R\in\mathcal{R}}f_R(S_B|S_A)\right)\cdot n/\theta$ could be an
accurate estimate of the influence spread of $S_B$ given $S_A$, i.e.,
$\sigma(S_B|S_A)$.

\begin{lemma}\label{lemma: theta requirement}
Suppose we are given a set $\mathcal{R}$ of $\theta$ random RAPG instances,
where $\theta$ satisfies
\begin{equation}
\theta \geq (8+2\epsilon)n\cdot
\frac{\ell \ln n + \ln \binom{n}{k}+\ln 2}{OPT\cdot \epsilon^2}.
\label{eq: theta requirement}
\end{equation}
Then, with probability at least $1-n^{-\ell}$, 
\begin{equation}
\left|\frac{n}{\theta}\cdot F_\mathcal{R}(S_B|S_A)-\sigma(S_B|S_A)\right|
	< \frac{\epsilon}{2}OPT
\label{eq: theta estimate accuracy}
\end{equation}
holds for all $S_B\subseteq V\backslash S_A$ with $k$ nodes.
\end{lemma}

\begin{proof}
First, let $S_B$ be a given seed set with $k$ nodes.  Let
$\mu=\mathbb{E}[f_R(S_B|S_A)]$,
$F_\mathcal{R}(S_B|S_A)=\sum_{R\in\mathcal{R}}f_R(S_B|S_A)$ can be regarded as
the sum of $\theta$ i.i.d. variables with a mean $\mu$.  By Lemma \ref{lemma:
EISB}, we have $\mu=\sigma(S_B|S_A)/n\leq OPT/n$.  Thus, by Chernoff-Hoeffding
bound,
\begin{align*}
&\Pr\left[\left|\frac{n}{\theta}\cdot F_\mathcal{R}(S_B|S_A)-
	\sigma(S_B|S_A)\right| \geq \frac{\epsilon OPT}{2}\right ] \\
=& \Pr\left[\left|F_\mathcal{R}(S_B|S_A)-\theta\cdot \mu\right|
	\geq \frac{\epsilon OPT}{2n\mu}\cdot \theta\mu\right] \\
\leq & 2\exp\left(-\frac{(\frac{\epsilon OPT}{2n\mu})^2}
	{2+\frac{\epsilon OPT}{2n\mu}}\!\cdot\!\theta\mu\right)
\!=\! 2\exp\left(-\frac{\epsilon^2 OPT^2}{8n^2\mu+2\epsilon n OPT}
	\!\cdot\! \theta\right)\\
\leq & 2\exp\left(-\frac{\epsilon^2 OPT}{(8+2\epsilon)\cdot n}\cdot \theta\right)
\leq n^{-\ell}/\binom{n}{k}.
\end{align*}
The last step follows by Inequality (\ref{eq: theta requirement}).
There are at most $\binom{n}{k}$ node set $S_B\subseteq V\backslash S_A$ with 
$k$ nodes. By union bound, with probability at least 
$1-n^{-\ell}$, Inequality (\ref{eq: theta estimate accuracy})
holds for all $S_B\subseteq V\backslash S_A$ with $k$ nodes.
\end{proof}

For the value of $\theta$ in Algorithm \ref{algo: node selection},
we have the following theorem.

\begin{theorem}\label{theorem: theta requirement}
Given that $\theta$ satisfies Inequality (\ref{eq: theta requirement}),
\mbox{Algorithm \ref{algo: node selection}} returns a solution with 
$(1-1/e-\epsilon)$ approximation with probability at least $1-n^{-\ell}$.
\end{theorem}

\begin{proof}
Suppose we are given a set $\mathcal{R}$ of $\theta$ random RAPG instances
where $\theta$ satisfies Inequality (\ref{eq: theta requirement}).
Let $S_B$ be the set of nodes returned by Algorithm \ref{algo: node selection}
and let $S_B^*$ be the set that maximizes $F_\mathcal{R}(S_B^*|S_A)$.
As we are using a $(1-1/e)$ greedy approach to select $S_B$,
we have $F_\mathcal{R}(S_B|S_A)\geq (1-1/e)F_\mathcal{R}(S_B^*|S_A)$.

Let $S_B^{\text{opt}}$ be the optimum seed set for source $B$, i.e.,
the set of nodes that maximizes the influence spread of $B$.
We have $F_\mathcal{R}(S_B^{\text{opt}}|S_A)\leq F_\mathcal{R}(S_B^*|S_A)$.

By Lemma \ref{lemma: theta requirement},
with probability at least $1-n^{-\ell}$,
we have
$\sigma(S_B|S_A)\geq F_\mathcal{R}(S_B|S_A)\cdot n/\theta-OPT\cdot \epsilon /2$
holds simultaneously for all $S_B\subseteq V\backslash S_A$ with $k$ nodes.

Thus, we can conclude
\begin{align*}
\sigma(S_B|S_A) & \geq \frac{n}{\theta}\cdot F_\mathcal{R}(S_B|S_A) 
	- \frac{\epsilon}{2}OPT \\
& \hspace{-0.5in} \geq \frac{n}{\theta}\cdot (1-1/e)F_\mathcal{R}(S_B^*|S_A) 
	- \frac{\epsilon}{2}OPT \\
& \hspace{-0.5in} \geq \frac{n}{\theta}\cdot (1-1/e)F_\mathcal{R}(S_B^{\text{opt}}|S_A) 
	- \frac{\epsilon}{2}OPT \\
& \hspace{-0.5in} \geq (1\!-\!1/e)(OPT\!-\!\frac{\epsilon}{2}OPT)\!-\!\frac{\epsilon}{2}OPT 
 \geq (1\!-\!1/e\!-\!\epsilon)OPT,
\end{align*}
which completes the proof.
\end{proof}

By Theorem \ref{theorem: theta requirement}, let
$\lambda\!=\!(8+2\epsilon)n
(\ell \ln n + \ln \binom{n}{k}+\ln 2)/\epsilon^2$,
we know Algorithm \ref{algo: node selection} returns a $(1-1/e-\epsilon)$-approximate solution for any $\theta\geq \lambda/OPT$.

\subsection{Parameter Estimation}

The goal of our parameter estimation algorithm is to find a lower 
bound $LB_e$ of $OPT$ so that $\theta=\lambda/LB_\text{e}\geq \lambda/OPT$.
Here, the subscript ``e'' of $LB_\text{e}$ is short for ``estimated''.

\mynoindent\textbf{Lower bound of $OPT$.}
We first define graph $G'=(V,E')$ as 
a subgraph of $G$ with all edges pointing to $S_A$ removed,
i.e., $E'\!=\!\{e_{uv}|e_{uv}\in E, v\in V\backslash S_A\}$.
Let $m'=|E'|$.
Then, we define a probability distribution $\mathcal{V^+}$ over the 
nodes in $V\backslash S_A$, such that the probability mass for each node is 
proportional to its number of incoming neighbors in $G'$.
Suppose we take $k$ samples from $\mathcal{V}^+$ and use them to form 
a node set $S_B^+$ with duplicated nodes eliminated.
A natural lower bound of $OPT$ would be the expected influence spread of 
$S_B^+$ given the seeds for source $A$ is $S_A$, i.e., $\sigma(S_B^+|S_A)$.
Furthermore, any lower bound of $\sigma(S_B^+|S_A)$ is also a lower bound of $OPT$.
In the following lemma, we present a lower bound of $\sigma(S_B^+|S_A)$.

\begin{lemma}\label{lemma: lower bound OPT}
Let $R$ be a random RAPG instance and let 
$V_R'=\{u|u\in V_R\backslash S_A, f_R(\{u\}|S_A)=1\}$.
We define the width of $R$, denoted by $w(R)$, as the number of edges in $G$ 
pointing to nodes in $V_R'$. Then, we define
\begin{equation}
\alpha(R)=1-\left(1-\frac{w(R)}{m'}\right)^k.
\end{equation}
We have $n\cdot \mathbb{E}[\alpha(R)]\leq \sigma(S_B^+|S_A)$, where 
the expectation of $\mathbb{E}[\alpha(R)]$ is taken over the randomness of $R$.
\end{lemma}

\begin{proof}
Let $S_B^+$ be a set formed by $k$ samples from $\mathcal{V}^+$ 
with duplicated nodes eliminated and suppose we are given a random RAPG instance $R$.
Let $p_1(R)$ be the probability that $S_B^+$ overlaps with $V_R'$.
For any $S_B^+$, we have $f_R(S_B^+|S_A)\geq 0$ by definition of the scoring 
system.
Moreover, if $S_B^+$ overlaps with $V_R'$, we would have $f_R(S_B^+|S_A)=1$.
Hence, $p_1(R)\leq f_R(S_B^+|S_A)$ holds
and $n\cdot \mathbb{E}[p_1(R)]\leq n\cdot \mathbb{E}[f_R(S_B^+|S_A)]=\sigma(S_B^+|S_A)$ 
follows from Lemma \ref{lemma: EISB}.
Furthermore, suppose we randomly select $k$ edges 
from $E'$ and form a set $E^+$.
Let $p_2(R)$ be the probability that at least one edge in $E^+$ points to 
a node in $V_R'$.
It can be verified that $p_1(R)=p_2(R)$.
From the definition of $w(R)$, we have 
$p_2(R)=\alpha(R)=1-(1-w(R)/m')^k$.

Therefore, we can conclude that
\begin{equation*}
\mathbb{E}[\alpha(R)]=\mathbb{E}[p_2(R)]=\mathbb{E}[p_1(R)]
\leq \sigma(S_B^+|S_A)/n,
\end{equation*}
which completes the proof.
\end{proof}

Let $LB_\text{e}:=n\cdot \mathbb{E}[\alpha(R)]$.
Then, Lemma \ref{lemma: lower bound OPT} 
shows that $LB_\text{e}$ is a lower bound of $OPT$.

\mynoindent\textbf{Estimation of the lower bound.}
By Lemma \ref{lemma: lower bound OPT}, we can estimate $LB_\text{e}$ by first 
measuring $n\cdot \alpha(R)$ on a set of random RAPG instances and then 
take the average of the estimation.
By Chernoff-Hoeffding bound, to obtain an estimation of $LB_\text{e}$ within 
$\delta\in[0,1]$ relative error with probability at least $1-n^{-\ell}$, 
the number of measurements required is 
$\Omega(n\ell \log n \epsilon^{-2}/LB_\text{e})$.
The difficulty is that we usually have no prior knowledge about $LB_\text{e}$.
In \cite{tang2014influence}, Tang et al. provided an adaptive sampling 
approach which dynamically adjusts the number of measurements based on 
the observed sample value.
Suppose $S_A=\emptyset$, the lower bound $LB_\text{e}$ we want to estimate
equals to the lower bound of maximum influence spread estimated 
in \cite{tang2014influence}.
Hence, we apply Tang et al.'s approach directly and Algorithm \ref{algo: estimate LB} 
shows the pseudo-code that estimates $LB_\text{e}$.
\begin{algorithm}[htb]
\small
\caption{EstimateLB$(G,\ell)$~\cite{tang2014influence}}
\label{algo: estimate LB}
\begin{algorithmic}[1]
	\FOR{$i=1$ \TO $\log_{2}n-1$}
		\STATE Let $c_i=(6\ell \ln n+6\ln(\log_2n))\cdot 2^{i}$.
		\STATE Let $s_i=0$.
		\FOR {$j=1$ \TO $c_i$}
			\STATE Generate a random RAPG instance $R$ and calculate $\alpha(R)$.
			\STATE Update $s_i=s_i+\alpha(R)$.
		\ENDFOR
		\IF {$s_i > c_i/2^i$}
			\RETURN $LB_\text{e}^*=n\cdot s_i/(2\cdot c_i)$.
		\ENDIF
	\ENDFOR
	\RETURN $LB_\text{e}^*=1$.
\end{algorithmic}
\end{algorithm}

For Algorithm \ref{algo: estimate LB}, the theoretical analysis
in \cite{tang2014influence} can be applied directly 
and the following theorem holds.
For the proof of Theorem \ref{theorem: LB star}, we refer 
interested readers to \cite{tang2014influence}.

\begin{theorem}\label{theorem: LB star}
When $n\geq 2$ and $\ell\geq 1/2$, Algorithm \ref{algo: estimate LB} 
returns $LB_\text{e}^*\in[LB_\text{e}/4,OPT]$ with at least $1-n^{-\ell}$ 
probability, and has expected running time $O(\ell(m+n)\log n)$.
Furthermore, $\mathbb{E}[1/LB_\text{e}^*]<12/LB_\text{e}$.
\end{theorem}

\mynoindent\textbf{Running time of the node selection process.}
We have shown how to estimate a lower bound of $OPT$, now we analyze
\mbox{Algorithm \ref{algo: node selection}} assuming 
$\theta\!=\!\lambda/LB_\text{e}^*$.
From $\theta\!\geq\! \lambda/OPT$ and 
Theorem \ref{theorem: theta requirement},
we know \mbox{Algorithm \ref{algo: node selection}} returns a 
$(1-1/e-\epsilon)$-approximate solution with high probability. 
Now we analyze the running time of Algorithm \ref{algo: node selection}.
The running time of building $\theta$ random RAPG instances is 
$O(\theta\cdot \mathbb{E}[N_R])=O(\frac{\lambda}{LB_\text{e}^*}
\cdot \mathbb{E}[N_R])$ where $\mathbb{E}[N_R]$
is the expected number of random numbers generated 
for building a random RAPG instance.
The following lemma shows the relationshp between $LB_\text{e}$ and $\mathbb{E}[N_R]$.

\begin{lemma}\label{lemma: ERAND leq LB}
$LB_\text{e}\geq \frac{n}{m}\mathbb{E}[N_R]$.
\end{lemma}

\begin{proof}
For a given RAPG instance $R$, recall that $w(R)$ is defined as the number of 
edges in $G$ pointing to any node $u\in V$ such that $u\in V\backslash S_A$ and 
$f_R(\{u\}|S_A)=1$.
If we generate a random number for an edge $e_{uv}\in E$ during the 
generation of $R$, 
we know $v\in V\backslash S_A$ and $f_R(\{v\}|S_A)=1$.
Hence, the number of random number generated during the generation of $R$
is no more than $w(R)$ and we have $\mathbb{E}[N_R]\leq \mathbb{E}[w(R)]$.
Moreover, we can conclude that
\begin{align*}
\frac{n}{m} \mathbb{E}[N_R] \leq & n\cdot\mathbb{E}\left[\frac{w(R)}{m}\right]
\leq n\cdot\sum_R \left(\Pr(R)\cdot \frac{w(R)}{m'}\right) \\
\leq & n\cdot\sum_R \left(\Pr(R)\cdot \alpha(R)\right)
= n\cdot\mathbb{E}[\alpha(R)],
\end{align*}
which completes the proof.
\end{proof}

Based on Theorem \ref{theorem: LB star}
showing $\mathbb{E}[1/LB_\text{e}^{*}]=O(1/LB_\text{e})$
and Lemma \ref{lemma: ERAND leq LB} showing 
$LB_\text{e}\geq \mathbb{E}[N_R]\cdot n/m$,
we can conclude that
$\mathbb{E}\left[\mathbb{E}[N_R]/LB_\text{e}^*\right]=O(1+m/n)$.
Recall that the greedy selection process in Algorithm \ref{algo: node selection} 
has time complexity $O(kn + c\theta\cdot \mathbb{E}[N_R])$.
Let $\theta=\lambda/LB_\text{e}^{*}$, the total running time of 
Algorithm \ref{algo: node selection} becomes
$O\big(kn\!+\!c\lambda \mathbb{E}[N_R]/LB_\text{e}^{*}\big)
\!=\!O\left(c(\ell\!+\!k)(m\!+\!n)\log n/\epsilon^{2}\right)$.

\subsection{Parameter Refinement}

As discussed before, if the lower bound of $OPT$ is tight,
our algorithm will have a short running time.
The current lower bound $LB_\text{e}$ is no greater than the expected 
influence spread of a set of $k$ independent samples from $\mathcal{V^+}$, 
with duplicated eliminated. 
Hence, $LB_\text{e}$ is often much smaller than the $OPT$. 
To narrow the gaps between the lower bound we get in 
Algorithm \ref{algo: estimate LB} and $OPT$,
we use a greedy algorithm to find a seed set $S_B'$ based on the limited number 
of RAPG instances we have already generated in 
Algorithm \ref{algo: estimate LB}, and 
estimate the influence spread of $S_B'$ with a reasonable accuracy. 
Then, the intuition is that we can use a creditable lower bound of 
$\sigma(S_B'|S_A)$ or $LB_\text{e}^*$, whichever is larger, 
as the \textit{refined} bound.

Algorithm \ref{algo: refine LB} describes how to refine the lower bound.
\mbox{Line \ref{algo: refine LB: greedy s}-\ref{algo: refine LB: greedy e}}
uses the greedy approach to find a seed set $S_B'$ based on the RAPG instances 
generated in Algorithm \ref{algo: estimate LB}.
Intuitively, $S_B'$ should have a large influence spread 
when used as seed set for source $B$.
Line \ref{algo: refine LB: influence s}-\ref{algo: refine LB: influence e}
estimates the expected influence of $S_B'$, i.e. $\sigma(S_B'|S_A)$.
By Lemma \ref{lemma: EISB}, let $\mathcal{R''}$ be a set of RAPG instances, 
$F:=n\cdot \left(\sum_{R\in\mathcal{R''}}f_R(S_B'|S_A)\right)/|\mathcal{R''}|$ is 
an unbiased estimation of $\sigma(S_B'|S_A)$.
Algorithm \ref{algo: refine LB} generates a sufficiently large number of 
RAPG instances and put them into $\mathcal{R''}$ such that 
$F\leq (1+\epsilon')\sigma(S_B'|S_A)$ holds with high probability.
Then, with high probability, we have 
$F/(1+\epsilon')\leq\sigma(S_B'|S_A)\leq OPT$.
We use $LB_\text{r}=\max\{F/(1+\epsilon'),LB_\text{e}^*\}$ 
as the refined lower bound of $OPT$, which will be used to derive $\theta$ 
in Algorithm \ref{algo: node selection}.
The subscript ``r'' of $LB_\text{r}$ stands for ``refinement''.

\begin{algorithm}[htb]
\small
\caption{RefineLB$(G,k,S_A,LB_\text{e}^*,\epsilon,\ell)$}\label{algo: refine LB}
\begin{algorithmic}[1]
	\STATE Let $\mathcal{R'}$ be the set of RAPG instances generated 
				in Algorithm \ref{algo: estimate LB}.
	\STATE Let $MG_\mathcal{R'}(u)=\sum_{R\in \mathcal{R'}}f_R(\{u\}|S_A)$ 
	for all $u\in V\backslash S_A$.\label{algo: refine LB: greedy s}
	\STATE Initialize the seed set $S_B'=\emptyset$. 
	\FOR{$i=1$ \TO $k$}
		\STATE Identity the node $v_i\!\in\!V\backslash(S_A\cup S_B)$ 
		with largest $MG_\mathcal{R'}(v_i)$.
		\STATE Add $v_i$ to $S_B'$.
		\IF {$i<k$}
			\STATE Update $MG_\mathcal{R'}(u)$ as
					$\sum_{R\in\mathcal{R}}\Delta_R(u|S_A,S_B)$ 
					for all $u\in V\backslash(S_A \cup S_B)$.
		\ENDIF
	\ENDFOR \label{algo: refine LB: greedy e}
	\STATE $\epsilon'=5\cdot\sqrt[3]{\ell\cdot \epsilon^2/(\ell+k)}$ 
	\label{algo: refine LB: influence s}
	\STATE $\lambda'=(2+\epsilon')\ell n\ln n/\epsilon'^2$ 
	\STATE $\theta'=\lambda'/LB_\text{e}^*$
	\STATE Generate a set $\mathcal{R''}$ of $\theta'$ random RAPG instances.
	\STATE Let $F=n\cdot \left(\sum_{R\in\mathcal{R''}}f_R(S_B'|S_A)\right)/\theta'$.
	\label{algo: refine LB: influence e} \label{algo: refine LB: F def}
	\RETURN $LB_\text{r}=\max\{F/(1+\epsilon'),LB_\text{e}^*\}$
\end{algorithmic}
\end{algorithm}

\mynoindent\textbf{Theoretical analysis.}
We now prove that Algorithm \ref{algo: refine LB} returns 
$LB_\text{r}\in [LB_\text{e}^*,OPT]$ with a high probability.

\begin{lemma}\label{lemma: LB plus}
If $LB_\text{e}^*\in[LB_\text{e}/4,OPT]$, Algorithm \ref{algo: refine LB} 
returns $LB_\text{r}\in [LB_\text{e}^*,OPT]$ with at least $1-n^{-\ell}$ probability.
\end{lemma}

\begin{proof}
As $LB_\text{r}=\max\{F/(1+\epsilon'),LB_\text{e}^*\}$ and 
$LB_\text{e}^*\leq OPT$, it is suffice to show $F/(1+\epsilon')\leq OPT$ 
holds with probability at least $1-n^{-\ell}$.
By Line \ref{algo: refine LB: F def} in Algorithm \ref{algo: refine LB}, 
we know $F/(1+\epsilon')\leq OPT$ if and only if
$\sum_{R\in\mathcal{R''}}f_R(S_B'|S_A)\leq OPT\cdot \theta'(1+\epsilon')/n$.
Let $\mu=\mathbb{E}[f_R(S_B'|S_A)]$, by Lemma \ref{lemma: EISB}, 
we have $n\mu=\sigma(S_B'|S_A)\leq OPT$.
Since $\sum_{R\in\mathcal{R''}}f_R(S_B'|S_A)$ is the summation of 
$|\mathcal{R}''|=\theta'$ i.i.d. random variable with mean $\mu$, 
By $n\mu\leq OPT$, $LB_\text{e}^*\leq OPT$ and Chernoff bound,
\begin{align*}
&\Pr\left[\sum_{R\in\mathcal{R''}}f_R(S_B'|S_A) 
				\geq \frac{OPT(1+\epsilon')}{n}\cdot \theta'\right] \\
\leq &\Pr\left[\sum_{R\in\mathcal{R''}}f_R(S_B'|S_A)-\mu\theta' 
				\geq \frac{OPT\epsilon'}{n\mu}\cdot \mu\theta'\right]\\
\leq & \exp\left(-\frac{\left(\frac{OPT\epsilon'}{n\mu}\right)^2}
				{2+\frac{OPT\epsilon'}{n\mu}}\cdot \mu\theta'\right)
\!=\! \exp\left(-\frac{OPT^2\epsilon'^2}{2n^2\mu+OPT\epsilon'n}\cdot \theta'\right)\\
\leq & \exp\left(-\frac{OPT\epsilon'^2}{(2+\epsilon')n}
			\cdot \frac{\lambda'}{LB_\text{e}^*}\right)
\leq \exp\left(-\frac{\epsilon'^2\lambda'}{(2+\epsilon')n}\right)
\leq \frac{1}{n^\ell}.
\end{align*}
The last inequality holds since $\lambda'=(2+\epsilon')\ell n\ln n/\epsilon'^2$ 
and this completes the proof.
\end{proof}

\mynoindent\textbf{Time complexity.}
We now analyze the time complexity of Algorithm \ref{algo: refine LB}.
The running time of Line 
\ref{algo: refine LB: greedy s}-\ref{algo: refine LB: greedy e} 
depends on $|\mathcal{R'}|$.
Theorem \ref{theorem: LB star} shows that the expected running time of 
\mbox{Algorithm \ref{algo: estimate LB}}
is $O(\ell(m+n)\log n)$, 
which means that the total number of edges in all RAPG instances is 
at most $O(\ell(m+n)\log n)$.
Hence, in Lines \ref{algo: refine LB: greedy s}-\ref{algo: refine LB: greedy e}
of Algorithm \ref{algo: refine LB},
the running time for the initialization and update of 
$MG_\mathcal{R'}$ would be $O(c\ell(m+n)\log n)$.
And the running time of Line 
\ref{algo: refine LB: greedy s}-\ref{algo: refine LB: greedy e} 
would be $O(c\ell(m+n)\log n + kn) = O(c(\ell+k)(m+n)\log n)$.
The running time of the Line 
\ref{algo: refine LB: influence s}-\ref{algo: refine LB: influence e}
is $O\left(\mathbb{E}[\lambda'/LB_\text{e}^*]\cdot \mathbb{E}[N_R]\right)$,
because we generate $\lambda'/LB_\text{e}^*$ RAPG instances and 
the running time of calculating $f_R(S_B'|S_A)$ for an RAPG instance 
$R=(V_R,E_R)$ is linear with $|V_R|$.
As $\mathbb{E}[1/LB_\text{e}^*]=O(12/LB)$ holds from 
Theorem \ref{theorem: LB star} and $\frac{n}{m}\mathbb{E}[N_R]\leq LB_\text{e}$ 
holds from Lemma \ref{lemma: ERAND leq LB}, we can conclude that
\begin{align*}
O\left(\mathbb{E}\left[\frac{\lambda'}{LB_\text{e}^*}\right]\cdot \mathbb{E}[N_R]\right)
&=O\left(\frac{\lambda'}{LB_\text{e}}\cdot \mathbb{E}[N_R]\right)\\
&\hspace{-1.2in}
=O\left(\frac{\lambda'}{LB_\text{e}}\cdot \left(1+\frac{m}{n}\right)LB_\text{e}\right)
=O(\ell (m+n)\log n/\epsilon'^2).
\end{align*}

To make sure that Algorithm \ref{algo: refine LB} has the same time 
complexity as Algorithm \ref{algo: node selection}, 
the value of $\epsilon'$ must satisfy 
$\epsilon'\geq \sqrt{\ell/(c(\ell+k))}\cdot \epsilon$.
In TIM/TIM$^+$~\cite{tang2014influence}
that returns approximation solution for single source influence 
maximization problem under the IC model,
Tang et al. set $\epsilon'=5\sqrt[3]{\ell\cdot\epsilon^2/(k+\ell)}$ 
for any $\epsilon\leq 1$.
Note that for a special case of the \textit{General Competitive Independent 
Cascade} model where $S_A=\emptyset$, the influence propagation model is 
actually the single source IC model.
Hence, we also set $\epsilon'=5\sqrt[3]{\ell\cdot\epsilon^2/(k+\ell)}$ 
for any $\epsilon\leq 1$.
Note that $c\geq 1$, it could be verified that 
$5\sqrt[3]{\ell\cdot\epsilon^2/(k+\ell)}\geq \sqrt{\ell/(c(\ell+k))}\epsilon$ 
holds for any $\epsilon\leq 1$.
Based on Lemma \ref{lemma: LB plus} and the time complexity analysis above, 
we have the following Theorem.
\begin{theorem}
Given $\mathbb{E}[1/LB_\text{e}^*]=O(12/LB_\text{e})$ and 
$LB_\text{e}^*\in[LB_\text{e}/4,OPT]$, Algorithm \ref{algo: refine LB} 
returns $LB_\text{r}\in [LB_\text{e}^*,OPT]$ with at least $1-n^{-\ell}$ 
probability and runs in $O(c(m+n)(\ell+k)\log n/\epsilon^{2})$ expected time.
\end{theorem}

If $LB_\text{r}\geq LB_\text{e}^*$, let $\theta=\lambda/LB_\text{r}$,
the total running time of Algorithm \ref{algo: node selection} is still
$O(c(\ell+k)(m+n)\log n/\epsilon^{2})$.

\subsection{TCIM as a whole}

Now we are in the position to put Algorithm \ref{algo: node
selection}-\ref{algo: refine LB} together and present the complete TCIM
algorithm.  Given a network $G$, the seed set $S_A$ for the source $A$ together
with parametric values $k$, $\ell$ and $\epsilon$, TCIM returns a
$(1-1/e-\epsilon)$ solution with probability at least $1-n^{-\ell}$.  First,
Algorithm \ref{algo: estimate LB} returns the estimated lower bound of $OPT$,
denoted by $LB_\text{e}^*$.  Then, we feed $LB_\text{e}^*$ to  Algorithm
\ref{algo: refine LB} and get a refined lower bound $LB_\text{r}$.  Finally,
Algorithm \ref{algo: node selection} returns a set $S_B$ of $k$ seeds for source
$B$ based on $\theta=\lambda/LB_\text{r}$ random RAPG instances.  Algorithm
\ref{algo: all} describes the pseudo-code of TCIM as a whole.

\begin{algorithm}[htb]
\small
\caption{TCIM $(G,S_A,k,\ell,\epsilon)$}
\label{algo: all}
\begin{algorithmic}[1]
	\STATE $\ell'=\ell+\ln 3/\ln n$
	\STATE $LB_\text{e}^*=\text{EstimateLB}(G,\ell')$
	\STATE $LB_\text{r}=\text{RefineLB}(G,k,S_A,LB_\text{e}^*,\epsilon,\ell')$
	\STATE $\lambda=(8+2\epsilon)n\left(\ell'\ln n+
				\ln \binom{n}{k}+\ln 2\right)/\epsilon^{2}$
	\STATE $\theta=\lambda/LB_\text{r}$
	\STATE $S_B=\text{NodeSelection}(G,S_A,k,\theta)$
	\RETURN $S_B$
\end{algorithmic}
\end{algorithm}

We use $\ell'=\ell+\ln 3/\ln n$ as the input parameter value of $\ell$ for 
Algorithm \ref{algo: node selection}-\ref{algo: refine LB}.
By setting this, Algorithm \ref{algo: node selection}-\ref{algo: refine LB} 
each fails with probability at most $n^{-\ell}/3$.
Hence, by union bound, TCIM succeeds in returning a 
$(1-1/e-\epsilon)$ approximation solution with probability at least $1-n^{-\ell}$.
Moreover, the total running time of TCIM is $O(c(\ell+k)(m+n)\log n/\epsilon^{2})$,
because Algorithm \ref{algo: node selection}-\ref{algo: refine LB}
each takes time at most $O(c(\ell+k)(m+n)\log n/\epsilon^{2})$.
In conclusion, we have the following theorem.

\begin{theorem}[TCIM]
TCIM returns $(1-1/e-\epsilon)$-approximate solution with probability 
at least $1-n^{-\ell}$. The time complexity is $O(c(\ell+k)(m+n)\log n/\epsilon^2)$.
\end{theorem}

\section{\bf Analyzing Various Propagation Models under GCIC} \label{sec: application}

In this section, we describe some special cases of the GCIC
model and provide detailed analysis about TCIM for these models.
To show the generality of the GCIC model, we use the 
\textit{Campaign-Oblivious Independent Cascade Model} 
in \cite{budak2011limiting},
the \textit{Distance-based model} and \textit{Wave propagation model} 
in \cite{carnes2007maximizing}
as specific propagation models.
For each specific model, we first briefly describe how the influence
propagates, give examples of score in a simple RAPG instance as shown in 
Figure \ref{fig: RAPG only}, and analyze the time complexity of 
the TCIM algorithm.

\begin{figure}[htb]
\centering
\includegraphics[width=0.17\textwidth]{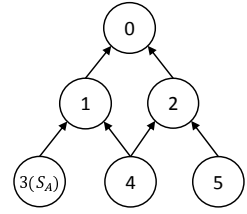}
\caption{Example of a random RAPG instance with $6$ nodes and $6$ edges.
Node $3$ is a seed of source $A$.}\label{fig: RAPG only}
\end{figure}

\subsection{Campaign-Oblivious Independent Cascade model}
Budak et al.~\cite{budak2011limiting} introduced the \textit{Campaign-Oblivious 
Independent Cascade model (COICM)} extending the single source IC model.
The influence propagation process 
starts with two sets of active nodes $S_A$ and $S_B$, and then unfolds 
in discrete steps.
At step $0$, nodes in $S_A$ (resp. $S_B$) are activated and are
in state $I_A$ (resp. $I_B$).
When a node $u$ first becomes activated in step $t$, it gets a single chance to 
activate each of its currently uninfluenced neighbor $v$ and succeeds with the
probability $p_{uv}$.
Budak et al. assumed that one source is prioritized over the other one 
in the propagation process, and nodes influence by the dominant source always 
attempt to influenced its uninfluenced neighbors first.
Here we assume that if there are two or more nodes 
trying to activate a node $v$ 
at a given time step, nodes in state $I_B$ 
(i.e., nodes influenced by source $B$) 
attempt first, which means  source $B$ is prioritized over source $A$.

\mynoindent\textbf{Examples of score.}
Suppose we are given seed sets $S_A$ and $S_B$ and a set of 
active edges $E_a$. In COICM, a node $u$ will be influenced by source $B$ 
if and only if $d_u(S_B,E_a)=d_u(S_A\cup S_B,E_a)$.
For the RAPG instance $R$ in Figure \ref{fig: RAPG only},
\mbox{$f_R(S_B|S_A)=1$} if $S_B\cap \{0,1,2,4,5\}\neq \emptyset$ and 
\mbox{$f_R(S_B|S_A)=0$} otherwise.

\mynoindent\textbf{Analysis of TCIM algorithm.}
Recall that while analyzing the running time of TCIM, 
we assume that if we have a set $\mathcal{R}$ of $\theta$ RAPG instances,
the time complexity for the initialization and update of the ``marginal gain 
vector'' $MG_\mathcal{R}$ is $O(c\theta\cdot \mathbb{E}\left[|E_R|\right])$.
We now show that $c=O(1)$ for COICM.
Suppose we are selecting nodes based on a set $\mathcal{R}$ of $\theta$ RAPG 
instances.
The initialization of $MG_\mathcal{R}$ takes time $O(\theta\cdot \mathbb{E}\left[|E_R|\right])$ 
as for any RAPG instance $R$, we have $f_R(\{u\}|S_A)=1$ for all 
$u\in V_R\backslash S_A$ and $f_R(\{u\}|S_A)=0$ otherwise.
Suppose in one iteration, we add a node $v_i$ to the set $S_B'$ and 
obtain a new seed set $S_B=S_B'\cup \{v_i\}$.
Recall that we define 
\mbox{$\mathcal{R'}=\{R|R\in\mathcal{R},\Delta_R(v_i|S_A,S_B')\}$}
in the greedy approach.
For every RAPG instance $R\in \mathcal{R'}$ and for all 
$u\in V\backslash(S_A\cup S_B)$, we would have $\Delta_R(u|S_A,S_B)=0$ 
and $\Delta_R(u|S_A,S_B')=1$  and hence we need to update 
$MG_\mathcal{R}(u)$ correspondingly.
For each RAPG instance $R$, it appears in $\mathcal{R'}$ 
in at most one iteration. 
Hence, the total time complexity of the initialization and update of the 
``marginal gain vector'' takes time $O(\theta \cdot \mathbb{E}\left[|E_R|\right])$.
It follows that the running time of TCIM is $O((\ell+k)(m+n)\log n/\epsilon^2)$.

\subsection{Distance-based Model}

Carnes et al. proposed the \textit{Distance-based model} 
in \cite{carnes2007maximizing}.
The idea is that a consumer is more likely to be influenced by the early 
adopters if their distance in the network is small.
The model governs the diffusion of source $A$ and $B$ given the initial 
adopters for each source and a set $E_a\subseteq E$ of active edges.
Let $d_u(E_a, S_A\cup S_B)$ be the shortest distance from $u$ to $S_A\cup S_B$ 
along edges in $E_a$
and let $d_u(E_a, S_A\cup S_B)=+\infty$ 
if there are no paths from $u$ to any node in $S_A\cup S_B$.
For any set $S\subseteq V$, we define $h_u(S, d_u(E_a, S_A\cup S_B))$ as the 
number of nodes in $S$ at distance $d_u(E_a, S_A\cup S_B)$ from $u$ along 
edges in $E_a$.
Given $S_A$, $S_B$ and a set of active edge $E_a$, the probability that 
node $u$ will be influenced by source $B$ is
\begin{equation}
\frac{h_u(S_B, d_u(E_a, S_A\!\cup\! S_B))}
{h_u(S_A\!\cup\! S_B, d_u(E_a, S_A\!\cup\! S_B))}.
\end{equation}
Thus, the expected influence of $S_B$ is
\begin{equation}
\sigma(S_B|S_A)\!=\!\mathbb{E}\left[\sum_{u\in V}
	\frac{h_u(S_B, d_u(E_a, S_A\!\cup\! S_B))}
	{h_u(S_A\!\cup\! S_B, d_u(E_a, S_A\!\cup\! S_B))}
\right],
\end{equation}
where the expectation is taken over the randomness of $E_a$.

\mynoindent\textbf{Examples of the score.}
Suppose we are given a random RAPG instance $R$ shown in 
Figure \ref{fig: RAPG only}.
If $S_B\cap \{0,1,2\}\neq \emptyset$, we would have $f_R(S_B|S_A)=1$.
Suppose $S_B=\{4,5\}$, we have $d_0(E_R, S_A\cup S_B)=2$, $h_0(S_B,2)=2$ 
and $h_0(S_A\cup S_B,2)=3$.
Hence the probability that node $0$ will be influenced by source $B$ 
is $\frac{2}{3}$ and we have $f_R(S_B=\{4,5\}|S_A)=\frac{2}{3}$.
For $S_B=\{4\}$ or $S_B=\{5\}$, one can verify that $f_R(S_B|S_A)=\frac{1}{2}$.

\mynoindent\textbf{Analysis of TCIM algorithm.}
We now show that $c=O(k)$ for the Distance-based Model.
In the implementation of TCIM under the Distance-based Model, 
for each RAPG instance $R=(V_R,E_R)$ with ``root'' $r$, we keep 
$d_R(v, r)$ for all $v\in V_R$ and $d_R(S_A, r)$ in the memory.
Moreover, we keep track of the value $h_r(S_A\cup S_B,d_R(S_A, r))$
and $h_r(S_B,d_R(S_A, r))$ for current $S_B$ and put them inside the memory.
Then, for any given RAPG instance $R=(V_R,E_R)$ and a node 
$u\in V_R\backslash(S_A\cup S_B)$, we have 
\begin{equation*}
f_R(S_B\cup \{u\}|S_A)
=\frac{h_r(S_B,d_R(S_A, r))+1}{h_r(S_A\cup S_B,d_R(S_A, r))+1}
\end{equation*}
if $d_R(u,r)=d_R(S_A,r)$ and $f_R(S_B\cup \{u\}|S_A)=1$ otherwise.
In each iteration, for each RAPG instance $R$, the update of 
$h_r(S_A\cup S_B,d_R(S_A, r))$ and $h_r(S_B,d_R(S_A, r))$ 
after expanding previous seed set $S_B$ by adding a node could be done in $O(1)$.
Moreover, for any $R$ and $u\in V_R\backslash(S_A\cup S_B)$,
the evaluation of $\Delta_R(u|S_A,S_B)$
could also be done in $O(1)$.
There are $O(\theta)$ RAPG instances with the total number of nodes 
being $O(\theta\cdot \mathbb{E}\left[|E_R|\right])$.
Hence, in $k$ iterations, it takes $O(k\theta\cdot \mathbb{E}\left[|E_R|\right])$ in total
to initialize and update the marginal gain vector.
Substituting $c$ with $O(k)$ in $O(c(\ell+k)(m+n)\log n/\epsilon^2)$,
the running time of the TCIM algorithm is $O(k(\ell+k)(m+n)\log n/\epsilon^2)$.

\subsection{Wave Propagation Model}
Carnes et al. also proposed the Wave Propagation
Model in \cite{carnes2007maximizing} extending the single source IC model.
Suppose we are given $S_A$, $S_B$ and a set of active edges $E_a$, we denote 
$p(u|S_A,S_B,E_a)$ as the probability that node $u$ gets influenced by source $B$.
We also let $d_{E_a}(S_A\cup S_B,u)$ be the shortest distance from seed nodes 
to $u$ through edges in $E_a$.
Let $N_u$ be the set of neighbors of $u$ whose shortest distance 
from seed nodes through edges in $E_a$ is $d_{E_a}(S_A\cup S_B,u)-1$.
Then, Carnes et al. \cite{carnes2007maximizing} defines 
\begin{equation}
p(u|S_A,S_B,E_a)=\frac{\sum_{v\in N_u} p(v|S_A,S_B,E_a)}{|N_u|}.
\end{equation}
The expected number of nodes $S_B$ can influence given $S_A$ is
\begin{equation}
\sigma(S_B|S_A)=\mathbb{E}\left[\sum\limits_{v\in V} p(v|S_A,S_B,E_a)\right],
\end{equation}
where the expectation is taken over the randomness of $E_a$.

\mynoindent\textbf{Examples of score.}
For a random RAPG instance $R$ shown in Figure \ref{fig: RAPG only}, 
as for the Distance-based Model, we have $f_R(S_B|S_A)=1$ if 
$S_B\cap \{0,1,2\}\neq \emptyset$.
Suppose $S_B=\{4\}$, source $B$ would influence node $4$ and $2$ with 
probability $1$, influence node $1$ with probability $\frac{1}{2}$ and 
influence node $0$ with probability $\frac{3}{4}$.
Hence, $f_R(\{4\}|S_A)=\frac{3}{4}$.
Suppose $S_B=\{5\}$, source $B$ would influence node $5$ and $2$ with 
probability $1$, influence node $0$ with probability $\frac{1}{2}$.
Hence, $f_R(\{5\}|S_A)=\frac{1}{2}$.
Moreover, one can verify that $f_R(\{4,5\}|S_A)=\frac{3}{4}$.

\mynoindent\textbf{Analysis of TCIM algorithm.}
We now show that for a greedy approach based on a set of $\theta$ random RAPG 
instances, it takes $O(kn\cdot \theta\cdot \mathbb{E}[N_R])$ in total to initialize and 
update the ``marginal gain vector''.
In each iteration of the greedy approach, for each RAPG instance 
$R=(V_R,E_R)$ and each node $u\in V_R\backslash(S_A\cup S_B)$,
it takes $O(|E_R|)\leq O(\mathbb{E}[N_R])$ to update the marginal gain vector.
Since there are $\theta$ RAPG instances each having at most $n$ 
nodes and the greedy approach runs in $k$ iteration,
it takes at most $O(kn\cdot \theta\cdot \mathbb{E}[N_R])$ in total 
to initialize and update the marginal gain vector.
Substituting $c=O(kn)$ into $O(c(\ell+k)n(m+n)\log n/\epsilon^2)$,
we can conclude that the running time of TCIM is 
$O(k(\ell+k)n(m+n)\log n/\epsilon^2)$.

\section{\bf Comparison with the Greedy Algorithm} \label{sec: comparison}

In this section, we compare TCIM to the greedy approach with Monte-Carlo method.
We denote the greedy algorithm as \textit{GreedyMC}, and it works as follows.
The seed set $S_B$ is set to be empty initially and
the greedy selection approach runs in $k$ iterations.
In the $i$-th iteration, \textit{GreedyMC} identifies a node 
$v_i\in V\backslash(S_A\cup S_B)$ that maximizes the marginal gain of
influence spread of source $B$, i.e., maximizes
$\sigma(S_B\cup\{v_i\}|S_A)-\sigma(S_B|S_A)$, and put it into $S_B$.
Every estimation of the marginal gain is done by $r$ Monte-Carlo simulations.
Hence, \textit{GreedyMC} runs in
at least $O(kmnr)$ time.

In \cite{tang2014influence}, Tang et al. provided the lower bound of 
$r$ that ensures the $(1-1/e-\epsilon)$ approximation ratio of this method
for single source influence maximization problem.
We extend their analysis on \textit{GreedyMC} and give the following theorem.

\begin{theorem}\label{th: monte carlo r}
For the Competitive Influence Maximization problem,
GreedyMC returns a $(1-1/e-\epsilon)$-approximate solution with at least
$1\!-\!n^{-\ell}$ probability, if 
\begin{equation}\label{eq: monte carlo r}
r\geq (8k^2+2k\epsilon)\cdot n
	\cdot \frac{(\ell+1)\ln n + \ln k}{\epsilon^2\cdot OPT}.
\end{equation}
\end{theorem}

\begin{proof}
Let $S_B$ be any node set that contains at most $k$ nodes in $V\backslash S_A$
and let $\sigma'(S_B|S_A)$ be the estimation of $\sigma(S_B|S_A)$ computed by 
$r$ Monte-Carlo simulations.
Then, $r\sigma'(S_B|S_A)$ can be regarded as the sum of $r$ i.i.d.
random variable bounded in $[0,1]$ with the mean value 
$\sigma(S_B|S_A)$.
By Chernoff-Hoeffding bound, if $r$ satisfies Ineq. (\ref{eq: monte carlo r}),
it could be verified that
$\Pr\left[|\sigma'(S_B|S_A)-\sigma(S_B|S_A)|
	>\frac{\epsilon}{2k}OPT\right]$
is at least $k^{-1}\cdot n^{-(\ell+1)}$.
Given $G$ and $k$, \textit{GreedyMC} considers at most $kn$ node sets 
with sizes at most $k$.
Applying the union bound, with probability at least $1-n^{-\ell}$, we have
\begin{equation}\label{eq: monte carlo accuracy}
|\sigma'(S_B|S_A)-\sigma(S_B|S_A)|>\frac{\epsilon}{2k}\cdot OPT
\end{equation}
holds for all sets $S_B$ considered by the greedy approach.
Under the assumption that $\sigma'(S_B|S_A)$ for all set $S_B$ considered 
by \textit{GreedyMC} satisfies Inequality (\ref{eq: monte carlo accuracy}),
\textit{GreedyMC} returns a $(1-1/e-\epsilon)$-approximate solution.
For the detailed proof of the accuracy of \textit{GreedyMC},
we refer interested readers to \cite{tang2014influence} (Proof of Lemma 10).
\end{proof}

\noindent
{\bf Remark:}
Suppose we know the exact value of $OPT$ and set $r$ to the smallest value 
satisfying Ineq. (\ref{eq: monte carlo r}), the time complexity of 
\textit{GreedyMC} would be 
$O(k^3\ell n^2m\log n\cdot \epsilon^{-2}/ OPT)$.
Given that $OPT\leq n$, the time complexity of \textit{GreedyMC}
is at least $O(k^3\ell nm\log n/\epsilon^{2})$.
Therefore, if $c\leq O(k^3)$, TCIM is much more efficient 
than \textit{GreedyMC}.
If $c=O(k^3n)$, the time complexity of TCIM is still competitive with
\textit{GreedyMC}.
Since we usually have no prior knowledge about $OPT$,
even if $c\geq O(k^3n)$, TCIM is still a better choice than 
the \textit{GreedyMC}.

\section{{\bf Experimental results}} \label{section: experiments}

Here, we present experimental results 
on three real-world networks
to demonstrate the effectiveness and efficiency of the TCIM framework.

\mynoindent\textbf{Datasets.}
Our datasets contain three real-world networks:
(i) A \textit{Facebook-like social network}
containing $1,899$ users and $20,296$ directed edges~\cite{opsahl2009clustering}.
(ii) The \textit{NetHEPT network}, an academic collaboration network
including $15,233$ nodes and $58,891$ undirected edges~\cite{chen2009efficient}.
(iii) An \textit{Epinions social network} of the who-trust-whom relationships 
from the consumer review site Epinions~\cite{richardson2003trust}.
The network contains $508,837$ directed ``trust'' relationships 
among $75,879$ users.
As the \textit{weighted IC} model in \cite{kempe2003maximizing},
for each edge $e_{uv}\in E$,
suppose the number of edges pointing to $v$ is $d_v^-$, 
we set 
$p_{uv}=1/d_v^-$.

\mynoindent\textbf{Propagation models.}
For each dataset listed above, 
we use the following propagation models:
the \textit{Campaign-Oblivious Independent Cascade Model (COICM)},
\textit{Distance-based model} and
\textit{Wave propagation model}
as described in Section \ref{sec: application}.

\mynoindent\textbf{Algorithms.}
We compare TCIM with two greedy algorithms and 
a previously proposed heuristic algorithm, they are:                         
\begin{itemize}
\item {\bf CELF}:
A efficient greedy approach based on a ``lazy-forward'' optimization 
technique~\cite{leskovec2007cost}. 
It exploits the monotone and submodularity of the object function
to accelerate the algorithm.

\item {\bf CELF++}:
A variation of CELF which further exploits the submodularity of the 
influence propagation models~\cite{goyal2011celf++}.
It avoids some unnecessary re-computations of marginal gains in future 
iterations at the cost of introducing more computation for each candidate 
seed set considered in the current iteration.

\item {\bf SingleDiscount}:
A simple degree discount heuristic initially proposed
for single source influence maximization problem~\cite{chen2009efficient}.
For the CIM problem, we adapt this heuristic method and select $k$ nodes iteratively.
In each iteration, for a given set $S_A$ and current $S_B$, we select 
a node $u$ such that it has the maximum number of outgoing edges 
targeting nodes not in $S_A\cup S_B$.
\end{itemize}

For the TCIM algorithm, let $\mathcal{R}$ be all RAPG instances generated
in Algorithm \ref{algo: node selection}
and let $S_B$ be the returned seed set for source $B$,
we report $n\cdot \left(\sum_{R\in\mathcal{R}}(f_R(S_B|S_A)\right)/|\mathcal{R}|$ 
as the estimation of $\sigma(S_B|S_A)$.
For other algorithms tested, we estimate the influence spread of the 
returned solution $S_B$ using $50,000$ Monte-Carlo simulations.
For each experiment, we run each algorithm three times and report the 
average results.

\mynoindent\textbf{Parametric values.}
For TCIM, the default parametric values are
$|S_A|=50$, $\epsilon=0.1$, $k=50$, $\ell=1$.
For CELF and CELF++, for each candidate seed set $S_B$ under consideration,
we run $r\!=\!10,000$ Monte-Carlo simulations to 
estimate the expected influence spread of $S_B$.
We set $r\!=\!10,000$ following the practice in 
literature (e.g., \cite{kempe2003maximizing})
One should note that
the value of $r$ required in all of our experiment is much larger than $10,000$
by Theorem \ref{th: monte carlo r}.
For each dataset, the seed set $S_A$ for source $A$ is returned by the TCIM 
algorithm with parametric values $S_A=\emptyset$, $\epsilon=50$ and $\ell=1$.

\noindent
{\bf Results on Facebook-like network:}
We first compare TCIM to CELF, CELF++ and the SingleDiscount heuristic
on the \textit{Facebook-like social network}.

Figure \ref{fig: fblike influence vs k} shows the expected influence spread
of $S_B$ selected by TCIM and other methods.
One can observe that the influence spread of $S_B$ returned by TCIM, CELF and 
CELF++ are comparable.
The expected influence spread of the seeds selected by SingleDiscount is 
slightly less than other methods.
Interestingly, there is no significant difference between 
the expected influence 
spread of the seeds returned by TCIM with $\epsilon=0.1$ and $\epsilon=0.5$,
which shows that the quality of solution does not degrade too quickly
with the increasing of $\epsilon$.

Figure \ref{fig: fblike time vs k} shows the running time of TCIM, CELF and
CELF++, with $k$ varying from $1$ to $50$.
Note that we did not show the running time of SingleDiscount
because SingleDiscountit is a 
heuristic method and the expected influence spread 
of the seeds returned is inferior to the influence spread of
the seeds returned by the other three algorithms.
Figure \ref{fig: fblike time vs k} shows that among three 
influence propagation models,
as compared to CELF and CELF++,
TCIM runs {\em two to three orders of magnitude faster} if $\epsilon=0.1$
and {\em three to four orders of magnitude faster} when $\epsilon=0.5$.
CELF and CELF++ have similar running time because
most time is spent to select the first seed node for source $B$
and CELF++ differs from CELF starting from the selection of the second seed.

\begin{figure}[htb]
\centering
\includegraphics[width=1\columnwidth]{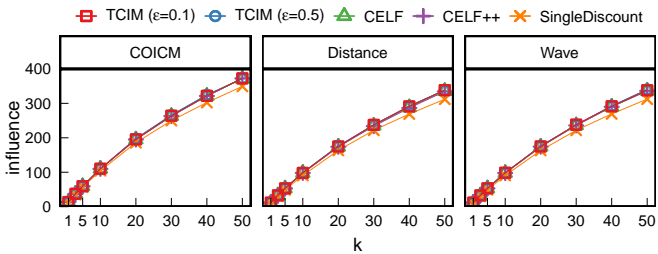}
  \caption{Results on the \textit{Facebook-like network}:
	Influence versus $k$ 
	under
	\textit{Campaign-Oblivious Independent Cascade Model (COCIM)},
	\textit{Distance-based model} and \textit{Wave propagation model}.
	($|S_A|=50$, $\ell=1$)}
  \label{fig: fblike influence vs k}
\end{figure}

\begin{figure}[htb]
\centering
\includegraphics[width=1\columnwidth]{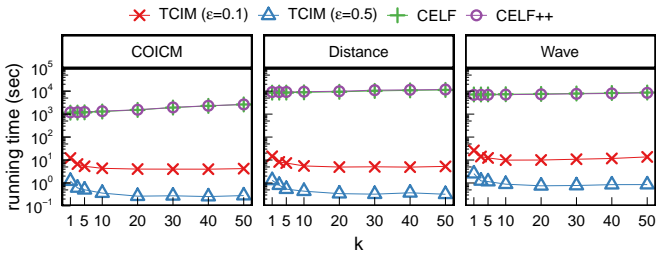}
  \caption{Results on the \textit{Facebook-like network}:
	Running time versus $k$ 
	under
	\textit{Campaign-Oblivious Independent Cascade Model (COCIM)},
	\textit{Distance-based model} and \textit{Wave propagation model}.
	($|S_A|=50$, $\ell=1$)}
  \label{fig: fblike time vs k}
\end{figure}

\noindent
{\bf Results on large networks:}
For \textit{NetHEPT} and \textit{Epinion},
we experiment by varying $k$, $|S_A|$ and $\epsilon$ to demonstrate the 
efficiency and effectiveness of the TCIM.
We compare the influence spread of TCIM to SingleDiscount heuristic only,
since CELF and CELF++ do not scale well on larger datasets.

\begin{figure}[htb]
\centering
\captionsetup[subfigure]{oneside,margin={0.6cm,0cm}}
\subfloat[\textit{NetHEPT}]{
\includegraphics[width=0.45\columnwidth]{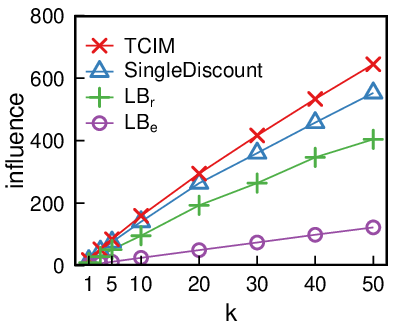}}
\subfloat[\textit{Epinion}]{
\includegraphics[width=0.45\columnwidth]{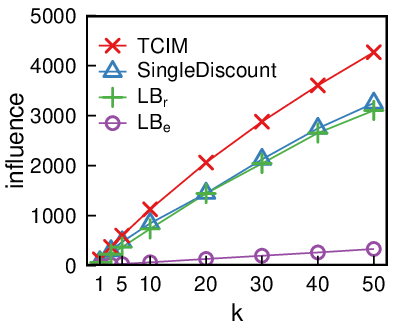}}
  \caption{Results on large datasets:
	Influence spreads versus $k$ under the \textit{Wave propagation model.}
	($|S_A|\!=\!50$, $\epsilon\!=\!0.1$, $\ell\!=\!1$)}
  \label{fig: large influence vs k}
\end{figure}

Figure \ref{fig: large influence vs k} shows the influence spread of 
the solution returned by TCIM and SingleDiscount, where the influence
propagation model is the \textit{Wave propagation model}.
We also show the
value of
$LB_\text{e}$ and $LB_\text{r}$ returned by the lower
bound estimation and refinement algorithm.
On both datasets, the expected influence of the seeds returned by TCIM
exceeds the expected influence of the seeds return by SingleDiscount.
Moreover, as in TIM/TIM$^+$~\cite{tang2014influence}, for every $k$, the 
lower bound $LB_\text{r}$ improved by Algorithm \ref{algo: refine LB}
is significant larger than the lower bound $LB_\text{e}$ returned by
Algorithm \ref{algo: estimate LB}.
When the influence propagation model is COICM or the 
\textit{Distance-based model}, the results are similar to that in 
Figure \ref{fig: large influence vs k}.

\begin{figure}[htb]
\centering
\captionsetup[subfigure]{oneside,margin={0.6cm,0cm}}
\subfloat[\textit{NetHEPT}]{
\includegraphics[width=0.45\columnwidth]{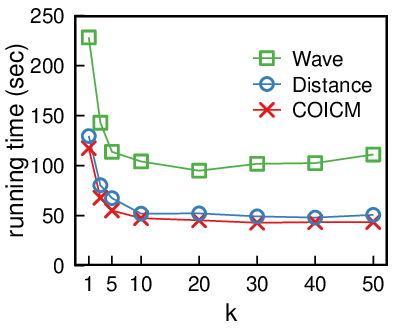}}
\subfloat[\textit{Epinion}]{
\includegraphics[width=0.45\columnwidth]{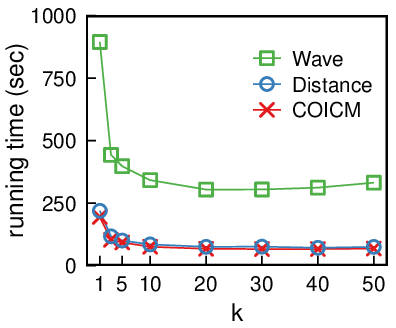}}
  \caption{Results on large datasets:
	Running time versus $k$ under 
	\textit{Campaign-Oblivious Independent Cascade Model (COCIM)},
	\textit{Distance-based model} and \textit{Wave propagation model}.
	($|S_A|=50$, $\epsilon=0.1$, $\ell=1$)}
  \label{fig: large time vs k}
\end{figure}

Figure \ref{fig: large time vs k} shows the running time of TCIM,
with $k$ varying from $1$ to $50$.
As in \cite{tang2014influence}, for every influence propagation model,
when $k\!=\!1$, the running time of TCIM is the largest.
With the increase of $k$, the running time tends to drop first,
and it may increase slowly after $k$ reaches a certain number.
This is because the running time of TCIM is mainly related to the number 
of RAPG instances generated in Algorithm \ref{algo: node selection},
which is $\theta=\lambda/LB_\text{r}$.
When $k$ is small, $LB_\text{r}$ is also small as $OPT$ is small.
With the increase of $k$, if $LB_\text{r}$ increases faster than 
the decrease of $\lambda$, $\theta$ decreases and the running time of 
TCIM also tends to decrease.
From Figure \ref{fig: large time vs k}, we see that TCIM is especially
efficient when $k$ is large.
Moreover, for every $k$, among three models, the running time of TCIM
based on the
\textit{Campaign-Obilivous Independent Cascade Model} is the smallest
while the running time of TCIM based on the \textit{Wave propagation model}
is the largest.
This is consistent with the analysis of the running of TCIM in 
Section \ref{sec: application}.

\begin{figure}
\centering
\captionsetup[subfigure]{oneside,margin={0.6cm,0cm}}
\subfloat[\textit{NetHEPT}]{
\includegraphics[width=0.45\columnwidth]{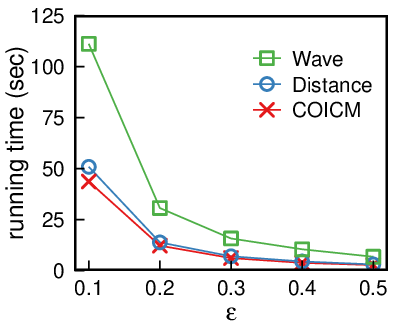}}
\subfloat[\textit{Epinion}]{
\includegraphics[width=0.45\columnwidth]{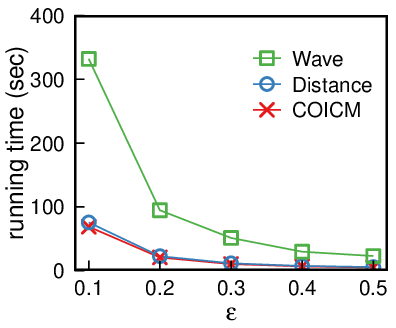}}
  \caption{Results on large datasets:
	Running time versus $\epsilon$ 
	under
	\textit{Campaign-Oblivious Independent Cascade Model (COCIM)},
	\textit{Distance-based model} and \textit{Wave propagation model}.
	($|S_A|=50$, $k=50$, $\ell=1$)}
  \label{fig: large time vs epsilon}
\end{figure}

Figure \ref{fig: large time vs epsilon} shows that the running time of 
TCIM decreases quickly with the increase of $\epsilon$,
which is consistent with its $O(c(\ell+k)(m+n)\log n/\epsilon^2)$ 
time complexity.
When $\epsilon\!=\!0.5$, TCIM finishes within 7 seconds for \textit{NetHEPT} 
dataset and finishes within 23 seconds for \textit{Epinion} dataset.
This implies that if we do not require a very tight approximation ratio,
we could use a larger $\epsilon$ as input and the performance of
TCIM could improve significantly.

\begin{figure}[htb]
\centering
\captionsetup[subfigure]{oneside,margin={0.6cm,0cm}}
\subfloat[\textit{NetHEPT}]{
\includegraphics[width=0.45\columnwidth]{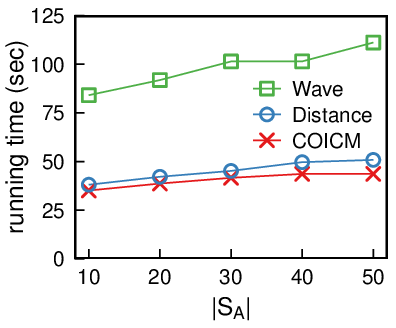}}
\subfloat[\textit{Epinion}]{
\includegraphics[width=0.45\columnwidth]{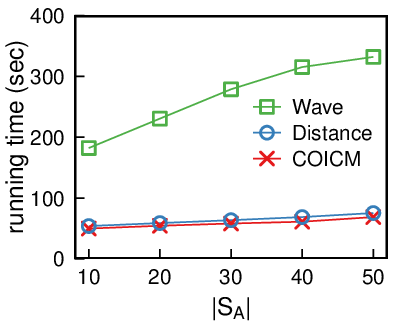}}
\caption{Results on large datasets:
	Running time versus $|S_A|$ 
	under
	\textit{Campaign-Oblivious Independent Cascade Model (COCIM)},
	\textit{Distance-based model} and \textit{Wave propagation model}.
	($k=50$, $\epsilon=0.1$, $\ell=1$)}
  \label{fig: large time vs sa}
\end{figure}

Figure \ref{fig: large time vs sa} shows the running time of TCIM
as a function of the seed-set size of source $A$.
For any given influence propagation model,
when $|S_A|$ increases, $OPT$ decreases and $LB_\text{r}$ tends to decrease.
As a result, the total number of RAPG instances required in the node
selection phase increases and consequently, the running time of TCIM
also increases.

\begin{figure}[ht]
\centering
\captionsetup[subfigure]{oneside,margin={0.3cm,0cm}}
\subfloat[\textit{NetHEPT}]{
\includegraphics[width=0.45\columnwidth]{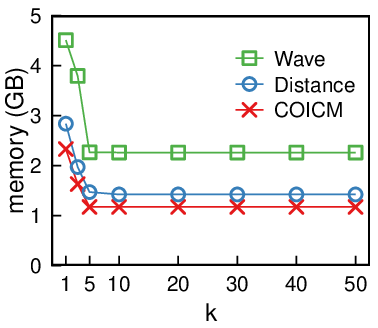}}
\subfloat[\textit{Epinion}]{
\includegraphics[width=0.45\columnwidth]{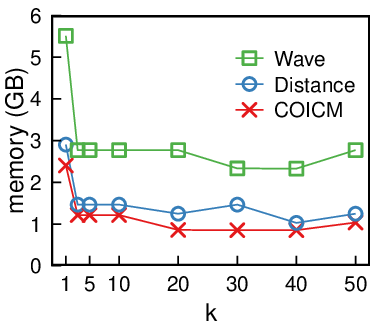}}
\caption{Results on large datasets:
	Memory consumption versus $k$ under 
	\textit{Campaign-Oblivious Independent Cascade Model (COCIM)},
	\textit{Distance-based model} and \textit{Wave propagation model}.
	($|S_A|=50$, $\epsilon=0.1$, $\ell=1$) }
  \label{fig: mem time vs k}
\end{figure}

Figure \ref{fig: mem time vs k} shows the memory consumption of TCIM
as a function of $k$.
For any $k$, TCIM based on
\textit{Campaign-Oblivious Independent Cascade Model}
consumes the least amount of memory because we only need to store the nodes
for each RAPG instance.
TCIM based on \textit{Wave propagation model} consumes the largest amount of memory 
because we need to store both the nodes and edges of each RAPG instance.
For the \textit{Distance-based model}, we do not need to store the edges
of RAPG instances, but need to store some other information for each RAPG
instance; therefore, the memory consumption is in the middle.
For all three propagation models and on both datasets, the memory requirement
drops when $k$ increases because the number of RAPG instances 
required tends to decrease.

\section{\bf Conclusion} \label{sec: conclusion}

In this work, we introduce a ``\textit{General Competitive Independent Cascade 
(GCIC)}'' model and define the ``\textit{Competitive Influence 
Maximization (CIM)}'' problem.
We then present a \textit{Two-phase Competitive Influence Maximization (TCIM)} 
framework to solve the CIM problem under GCIC model.
TCIM returns $(1-1/e-\epsilon)$-approximate solutions with probability at least 
$1-n^{-\ell}$ and has time complexity $O(c(\ell+k)(m+n)\log n/\epsilon^2)$, 
where $c$ depends on specific influence propagation model and may also depend 
on $k$ and graph $G$.
To the best of our knowledge, this is the first general algorithmic framework
for the \textit{Competitive Influence 
Maximization (CIM)} problem with both performance guarantee and practical running time.
We analyze TCIM under the \textit{Campaign-Oblivious Independent 
Cascade model} in \cite{budak2011limiting}, the \textit{Distance-based 
model} and the \textit{Wave propagation model} in \cite{carnes2007maximizing}.
And we show that, under these three models, 
the value of $c$ is $O(1)$, $O(k)$ and $O(kn)$ respectively.
We provide extensive experimental results to demonstrate the efficiency and 
effectiveness of TCIM.
The experimental results show that TCIM returns solutions comparable with those 
returned by the previous state-of-the-art greedy algorithms, but it runs 
{\em up to four orders of magnitute faster} than them.
In particular, when $k\!=\!50$, $\epsilon\!=\!0.1$ and 
$\ell\!=\!1$, given the set of $50$ 
nodes selected by the competitor, TCIM returns the solution within $6$ minutes 
for a dataset with $75,879$ nodes and $508,837$ directed edges.

\end{document}